\documentclass[%
  aps,           % APS journal style
  prd,           % Physical Review D
  nofootinbib,   % Avoid footnotes in bibliography
  floatfix,      % Fixes float positioning (tables/figures)
  superscriptaddress, % Author affiliations as superscripts
  groupedaddress,     % Group authors by affiliation
  showkeys,      % Show labels in draft (remove for final)
]{revtex4-2}     % or revtex4-2 (check journal's current preference)

\usepackage{graphicx}  % For figures
\usepackage{dcolumn}   % For aligned tables
\usepackage{bm}        % Bold math symbols
\usepackage{amsmath}   % Advanced math (amsfonts, amssymb included)
\usepackage{hyperref}  % Hyperlinks (load LAST)
\hypersetup{
  colorlinks=true,
  linkcolor=blue,
  urlcolor=blue,
  citecolor=blue,
  breaklinks=true,
}

\usepackage[normalem]{ulem} % For strike-through text (\sout)
\usepackage{amsfonts,amssymb}
\usepackage{epsfig}
\usepackage{wrapfig}
\usepackage{graphicx}
\usepackage[dvipsnames]{xcolor}
\usepackage{subfigure}
\usepackage[makeroom]{cancel}
\usepackage{float}
\usepackage{xcolor}
\usepackage{tensor}
\usepackage{cancel}

\usepackage{amsthm}
\usepackage{stmaryrd}

\usepackage{amscd,mathrsfs}
\usepackage{cases}
\usepackage{verbatim} 
\usepackage{epsf}

{
\theoremstyle{plain}
\newtheorem{theorem}{Theorem}
\newtheorem{assumption}{Assumption}
\theoremstyle{plain}
\newtheorem{remark}{Remark}[section]
}

\usepackage{appendix}

\newcommand{\be}{\begin{equation}}
\newcommand{\ee}{\end{equation}}
\newcommand{\bea}{\begin{eqnarray}}
\newcommand{\eea}{\end{eqnarray}}
\newcommand{\bml}{\begin{subequations}}
\newcommand{\eml}{\end{subequations}}

\newcommand{\bbm}{\begin{bmatrix}}
\newcommand{\ebm}{\end{bmatrix}}
\newcommand{\sech}{\mathrm{sech}}
\newcommand{\arcsinh}{\mathrm{arcsinh}}

\begin{document}

\title{Finite-time gradient blow-up and shock formation in Israel-Stewart theory: Bulk, shear, and diffusion regimes}

\today

\author{F\'abio S. Bemfica}
\affiliation{Department of Mathematics, Vanderbilt University, Nashville, TN, 37211, USA}
\affiliation{Escola de Ci\^encias e Tecnologia, Universidade Federal do Rio Grande do Norte, RN, 59072-970, Natal, Brazil}
\email{fabio.bemfica@ufrn.br}

\begin{abstract}
We present the first demonstration of finite-time gradient blow-ups in Israel-Stewart (IS) theories with 1+1D plane symmetry, mathematically showing the existence of smooth initial data that can evolve into shocks across three regimes: pure bulk viscosity, shear viscosity, and diffusion. Through numerical simulations of bulk-viscous fluids, we verify that these shocks satisfy Rankine-Hugoniot conditions, exhibit characteristic velocity crossing (Mach number obeys $\mathcal{M}_u > 1 > \mathcal{M}_d$), and maintain thermodynamic consistency, required for physical shocks. Our results reveal a crucial early-time dynamical phase --- previously unexplored in steady-state analyses --- where nonlinear effects dominate viscous damping, resolving the apparent impossibility of IS-type theories predicting shock formation. While restricted to simplified 1+1D systems with separate viscous effects, this work establishes foundational insights for shock formation in relativistic viscous hydrodynamics, highlighting critical challenges for extending to 3+1D systems or to a full IS theory where multiple nonlinear modes interact. The findings emphasize that both initial data structure and numerical methodology require careful consideration when studying shocks in relativistic viscous fluids.

\end{abstract}

%\pacs{Valid PACS appear here}
% PACS, the Physics and Astronomy Classification Scheme.
% Valid PACS numbers may be entered using the \verb+\pacs{#1} command.

% Use showkeys class option if keyword display desired

\maketitle

\section{Introduction}

Shock formation is a natural phenomenon in fluids\footnote{Here, the term \textit{fluid} encompasses liquids, gases, plasmas, or any medium that exhibits fluid-like behavior.} and is typically accompanied by an abrupt change in the fluid's physical variables such as density, temperature, and pressure. The first mathematical description of shock formation emerged from the analysis of solutions to the non-relativistic Euler equations for a compressible isothermal fluid in $1+1$ dimensions~\cite{Challis}. This work demonstrated the existence of a crossing point in the characteristic curves, where distinct smooth solutions of the equations coincide, leading to discontinuities in the fluid's propagation velocity and pressure. Since then, non-relativistic shock formation has been extensively studied and experimentally observed.

For high-energy phenomena, the relativistic Euler equations have served as the foundational framework for describing relativistic fluids. While these equations are limited to inviscid (zero-viscosity) fluids, they have been successfully applied to model relativistic fluid dynamics in diverse fields, including nuclear physics, astrophysics, and cosmology. In such contexts, shock formation is not only possible but often expected, particularly in extreme astrophysical events like supernovae~\cite{Koyama:1995rr}. These shocks arise during the collapse of massive stars and play a critical role in nucleosynthesis, influencing the production of heavy elements~\cite{Mekjian:1978zz}. Additionally, relativistic shocks are key drivers of astrophysical jets --- collimated outflows observed in systems such as active galactic nuclei and microquasars~\cite{Marti2019,Blandford2019}.

The study of shock formation in the relativistic Euler equations has been extended to $2+1$ and $3+1$ dimensions, with significant contributions from works such as Refs.~\cite{Alinhac1999,Alinhac2001a,Alinhac2001b,Christodoulou2007}\footnote{See also \cite{Speck} and the review \cite{Disconzi:2023rtt}}. These analyses provide deeper insights into the mathematical structure and physical implications of shocks in relativistic fluids, paving the way for further exploration in both theoretical and applied contexts.

Despite the success of the relativistic Euler equations in describing a wide range of physical phenomena, the inclusion of viscous corrections is often necessary. Relativistic viscous fluids play a crucial role in applications such as the quark-gluon plasma produced in experiments at the Relativistic Heavy Ion Collider and the Large Hadron Collider. Theoretical modeling must incorporate viscous corrections to accurately describe the experimental data from these collisions~\cite{Romatschke:2017ejr,Heinz:2013th}. Furthermore, recent numerical simulations~\cite{Alford:2017rxf,Shibata:2017xht,Shibata:2017jyf,Most:2021zvc,Most:2022yhe} of neutron star mergers using relativistic viscous hydrodynamic models strongly suggest that the resulting gravitational wave signals are influenced by viscous effects. However, shock formation in the viscous theories typically used to study these phenomena remains poorly understood.

Leading-edge simulations of heavy-ion collisions rely on Israel-Stewart (IS) theories of relativistic viscous fluids, first developed by Israel and Stewart~\cite{Israel:1976tn,Israel:1976efz,Israel1979}, with subsequent refinements in~\cite{Baier:2007ix,Denicol:2012cn}. The appeal of IS-type theories stems from their status as the first relativistic viscous fluid theories proven to be linearly stable and causal~\cite{Hiscock:1983zz,Olson:1990rzl}, in contrast to the acausal and unstable relativistic Navier-Stokes theories~\cite{Hiscock:1985zz}. Recent advances have further demonstrated the robustness of these theories: for a conformal fluid without heat flow described by the IS-type theory of~\cite{Denicol:2012cn}, it was shown in~\cite{Bemfica:2020xym} that the theory remains nonlinearly causal under certain field constraints. Additionally, in~\cite{Bemfica:2019cop}, an IS-type theory with only bulk viscosity was proven to be nonlinearly causal, with the associated partial differential equations forming a first-order symmetric hyperbolic system, whether or not the fluid is coupled to a dynamical gravitational field.

As of now, shock formation in IS-type theories remains poorly understood. The most comprehensive study to date is the work by Olson and Hiscock~\cite{OlsonHiscock1990} in $1+1$ dimensions, which analyzed late-time shock formation in the dissipation-dominated regime, where viscous and thermal effects (e.g., terms like $\eta\partial_x^2 u$) dominate over nonlinearities (e.g., $u\partial_x u$). Their analysis focused on asymptotic steady-state solutions, where dissipative mechanisms preserve the smoothness of the shock profile, precluding blow-up and yielding only finite-width shocks. A related study by Geroch and Lindblom~\cite{GEROCH1991394} also examined steady-state solutions in relativistic dissipative fluids in $1+1$ dimensions, demonstrating that causality and stability constraints impose strict bounds on shock structure in IS theories. However, neither work addressed the early-time regime, where nonlinearities may dominate and potentially lead to zero-width shocks (see Sec.~\ref{Sec:Discussion} for a detailed discussion of these regimes).  

A simplified IS-type theory with only bulk viscosity was studied in $3+1$ dimensions~\cite{Disconzi:2020ijk}. This work demonstrated that, for a class of smooth initial data, solutions must either develop singularities or become acausal within finite time. While this result establishes the inevitability of breakdown in certain regimes, the precise nature of the singularity (e.g., gradient blow-up, loss of causality, or shock formation) remains an open question. The contrast between these findings --- smooth late-time shocks versus potential early-time singularities --- highlights the need for a unified framework to describe shock formation across all dynamical regimes in viscous relativistic fluids.   

Recently, B\"arlin \cite{BARLIN2023103901} established that under certain assumptions --- notably the existence of genuinely nonlinear modes around an equilibrium reference state --- there exist smooth initial data that evolve into finite-time gradient blow-ups while keeping the dynamical variables bounded. This result is particularly significant because it applies to systems of first-order quasilinear partial differential equations in $1+1$ dimensions that include source terms, making it highly relevant for studying shock formation in IS-type theories. While the $1+1$ dimensional case is simpler than its $3+1$ dimensional counterpart, it provides crucial insights into the more complex higher-dimensional shock problem. In this work, we will leverage B\"arlin's results to demonstrate finite-time gradient blow-ups and shock formation in viscous relativistic fluids --- in particular IS-type theories.

These findings have important implications for relativistic heavy-ion collisions: smooth late-time shocks (as studied in \cite{OlsonHiscock1990}) may manifest in the collective flow observables of quark-gluon plasma, whereas finite-time singularities could relate to transient signals\footnote{Here, \textit{transient signals} refer to short-lived, non-equilibrium phenomena such as early-stage pressure anisotropies or rapid decorrelation patterns that cannot be described by hydrodynamics after thermalization. Such signals may carry information about the initial pre-equilibrium dynamics.} in the pre-equilibrium phase. The interplay between these regimes --- smooth dissipation-dominated shocks versus singularity-driven breakdowns --- could help explain the full evolution of the system from initial collision to final-state particles.

The paper is organized as follows. In the next section, we clarify why our results do not contradict the findings of Olson-Hiscock~\cite{OlsonHiscock1990} and Geroch-Lindblom~\cite{GEROCH1991394} on IS-type theories, but rather complement their analyses by focusing on the early-time regime where nonlinear modes dominate over viscous effects. Section~\ref{Math_framwork} establishes the mathematical framework to which B\"arlin's results~\cite{BARLIN2023103901} apply. We then derive conditions for the existence of smooth initial data that lead to finite-time gradient blow-ups in three distinct cases: for the $1+1$ IS equations with only bulk viscosity (Section~\ref{Bulk}), only shear viscosity (Section~\ref{Sec:Shear}), and only diffusion (Section~\ref{Sec:Diff}).

To analyze shock formation, Section~\ref{Numerical_Bulk} presents numerical solutions of the IS equations for the pure bulk viscosity case under two different initial conditions. The first setup generates a pair of shock waves, while the second produces a rarefaction wave coupled with a shock wave. In both scenarios, we examine the Rankine-Hugoniot (RH) conditions and the Mach number in the shock frame. Our results demonstrate that the RH conditions are satisfied within acceptable numerical precision, and we observe a transition from supersonic to subsonic Mach numbers --- a hallmark of irreversible processes and physical shocks. 

Conclusions and final remarks are provided in Section~\ref{Sec:Conclusion}. Appendix~\ref{Numerical_Scheme} details the numerical schemes employed, and Appendix~\ref{Bjorken} validates the numerical code using the semi-analytical Bjorken flow solution as a benchmark. 

Throughout this paper, we work in natural units where $c=\hbar=k_B=1$. We assume a flat spacetime with the Minkowski metric $g_{\mu\nu} = \text{diag}(-1, 1, 1, 1)$. By $1+1$ dimensions, we refer to a $3+1$ dimensional spacetime theory with $1+1$ dimensional symmetry, where all dynamical variables depend only on the coordinates $t$ (time) and $x$ (space). The fluid's flow velocity is given by 
\[u^\mu = (u^0, u, 0, 0) = \gamma(1, v, 0, 0),\]
where $\gamma = 1/\sqrt{1 - v^2}$ is the Lorentz factor and $v$ is the fluid's three-velocity. The projection tensor $\Delta^{\mu\nu} = u^\mu u^\nu + g^{\mu\nu}$ then has the following non-vanishing components: $\Delta^{00} = u^2 = \gamma^2 v^2,\, \Delta^{11} = u_0^2 = \gamma^2$, $\Delta^{10} = u^0 u = \gamma^2 v$, and $\Delta^{22}=\Delta^{33}=1$.

\section{Early-Time Nonlinear Dominance vs. Late-Time Viscous Dominance}
\label{Sec:Discussion}

To distinguish our work from the well-known studies of Olson-Hiscock \cite{OlsonHiscock1990} and Geroch-Lindblom \cite{GEROCH1991394} on IS theory, we clarify key conceptual differences.
\begin{enumerate}
\item \textbf{Late-Time Regime (Olson-Hiscock/Geroch-Lindblom)}: These works focus on \textit{steady-state solutions} of the IS equations, where dissipation dominates and time derivatives are negligible. The system reduces to
\[A^1 \partial_x \psi = B,\]
where $A^1\in\mathbb{R}^{n\times n}$ depends on the $n$ fluid variables, $\psi\in\mathbb{R}^{n\times 1}$ represents the fluid variables (e.g., shear stress, heat flow), and $B\in\mathbb{R}^{n\times 1}$ encodes damping terms like $-\pi^{\mu\nu}/\tau_\pi$. Crucially, steady-state solutions exclude strong shocks because the characteristic condition $\det[A^1] = 0$ at $v = v_{\text{char}}$ prevents admissible velocity jumps ($v_- > v_{\text{char}}$ to $v_+ < v_{\text{char}}$). This restricts solutions to small perturbations.

\item \textbf{Early-Time Regime (This Work)}: We study the \textit{full time-dependent system} $A^\alpha \partial_\alpha \psi=A^0 \partial_t \psi+A^1 \partial_x \psi = B$, where nonlinearities dominate. Causality requires the characteristic roots $\xi_0=\xi_0(\xi_1)$ of $\det[A^\alpha \xi_\alpha] = 0$ to define a real non-timelike co-vector $\xi_\mu$. For $1+1$-dimensional flow $u^\mu = \gamma(1, v, 0, 0)$, the characteristic determinant typically factors into polynomials like
\[(u^\alpha \xi_\alpha)^2 - v_{\text{char}}^2 \Delta^{\alpha\beta} \xi_\alpha \xi_\beta,\]
where $v_{\text{char}} \in [0,1]$ is a characteristic velocity in the fluid's rest frame. The full time-dependent system has an invertible matrix $A^0$ since the timelike co-vector $\xi_\mu=(1,0,0,0)$, which, from causality, cannot be one of the characteristics, yields $\det[A^0]\ne 0$. In the steady-state limit, $\det[A^1]$ corresponds to $\xi_\mu = (0,1,0,0)$ in the characteristic determinant, giving $\gamma^2(v^2 - v_{\text{char}}^2)=0$ when $v = v_{\text{char}}$, which explains why $\det[A^1] = 0$ occurs in this case and multi-valued solutions happen when $v$ crosses $v_{\text{char}}$.
\end{enumerate}

Therefore, the full analysis of shock formation requires studying the dynamical regime, which has not been fully addressed. Notably, the upstream/downstream velocity jump does not compromise the dynamical solution when causality holds, since $\det[A^0]\ne 0$. To illustrate the interplay between gradient blow-ups and shock formation, consider Burgers' equation with dissipation \cite{landau1987fluid,farlow1993partial}
\[\partial_t v + v\partial_x v = \beta\partial_x^2 v.\]
In the initial phase where nonlinearity dominates ($v\partial_x v \gg \beta\partial_x^2 v$), the gradient grows as $|\partial_x v|\sim (t-t_c)^{-1}$, with $t_c = -1/\partial_x v(0,x)_{\text{min}}$. This requires initial fluid compression ($\partial_x v(0,x)_{\text{min}}<0$) to ensure $t_c>0$. When dissipation dominates ($\partial_t v \approx \beta\partial_x^2 v$), gradients diffuse rather than grow. In the intermediate regime where both effects compete, the steady-state balance $v\partial_x v = \beta\partial_x^2 v$ yields tanh-like solutions with maximum gradient scaling as $|\partial_x v|_{\text{max}} \sim (\Delta v)^2/\beta$, where $\Delta v$ is the upstream-downstream velocity difference. This describes finite-width, smooth profiles --- i.e., weak shocks. This analysis clarifies that Olson-Hiscock studied late-time shocks where dissipation balances nonlinearity, precluding strong shocks due to causality constraints. For steep shock formation (gradient blow-ups), the condition $|v\partial_x v| \gg |\beta\partial_x^2 v|$ must hold, equivalent to $|\partial_x v(x,0)|^{-1} \gg L^2/\beta$ for some length scale $L$, ensuring $t_c \ll t_{\text{diss}}$ (where $t_{\text{diss}}$ is the dissipation timescale).

Now, let us consider the Israel-Stewart equation for the shear tensor
\[
\Delta^{\mu\nu}_{\alpha\beta}u^\lambda \partial_\lambda \pi^{\alpha\beta} + \frac{2\eta}{\tau_\pi}\sigma^{\mu\nu} = -\frac{\pi^{\mu\nu}}{\tau_\pi},
\]
where $\Delta^{\mu\nu}_{\alpha\beta} = \frac{1}{2}(\Delta^{\mu}_\alpha \Delta^\nu_\beta + \Delta^\mu_\beta\Delta^\nu_\alpha) - \frac{1}{3}\Delta^{\mu\nu}\Delta_{\alpha\beta}$,
$\sigma^{\mu\nu} = \Delta^{\mu\nu}_{\alpha\beta}\partial^{\alpha}u^{\beta}$, $\eta$ is the shear viscosity, and $\tau_\pi$ is the relaxation time. The term $-\pi^{\mu\nu}/\tau_\pi$ plays a crucial role as a relaxation term. Without it, the equation for $\pi^{\mu\nu}$ would not relax to equilibrium regardless of the value of $\eta$. However, when $\pi^{\mu\nu}/\tau_\pi$ is non-zero but sufficiently small, nonlinear modes may dominate, potentially leading to gradient blow-ups. This aligns with one of the key assumptions in \cite{BARLIN2023103901} for gradient blow-ups in strictly hyperbolic systems with source terms. Conversely, when $\pi^{\mu\nu}/\tau_\pi$ is not small at $t=0$ shocks may not be expected and the shear stress tensor may relax to its asymptotic solution
\[\pi^{\mu\nu} \approx -2\eta \sigma^{\mu\nu}.\]
In this case, the energy-momentum conservation equation $\partial_\nu T^{1\nu} = 0$ exhibits a competition between the nonlinear convection term $u\partial_x u$ from $u^\alpha\partial_\alpha u$ and the dissipative term $\partial_x^2 u$ from $\partial_\nu \pi^{1\nu} \approx -2\eta\partial_x^2 u$. This is precisely the regime where the dissipative Burgers equation analysis applies. In particular, the results of Olson-Hiscock/Geroch-Lindblom are based on this limiting case where nonlinear modes and dissipation balance each other.

\subsection*{Steady-State vs. Dynamic Shocks}

The key conclusions from the above discussion can be summarized as follows:
\begin{itemize}
    \item \textbf{Steady-state solutions} represent a special case where the shock profile becomes time-independent. In this regime, the requirement for the fluid velocity $v$ to cross the characteristic velocity $v_{\text{char}}$ arises from the ODE structure of the spatial equations. When $v = v_{\text{char}}$, the system becomes singular ($\det[A^1] = 0$), leading to multi-valued solutions. Consequently, this fundamentally limits the allowed velocity jumps and shock strengths.
    \item \textbf{Dynamic (time-dependent) shocks} are governed by the full partial differential equation (PDE) system, where the breakdown mechanism involves both spatial and temporal derivatives. The system may remain mathematically well-posed (e.g., in Sobolev spaces) while still developing finite-time gradient catastrophes. Nonlinear effects (wave steepening) can dominate over dissipation, enabling true shock formation. This occurs even when steady-state analysis would predict a breakdown of solutions for supersonic/subsonic crossing.
\end{itemize}
The works \cite{OlsonHiscock1990,GEROCH1991394} establish that when the system enters the viscosity-dominated phase, Israel-Stewart dissipation guarantees the late-time solution (post any potential blow-up) will become smooth. This provides an important connection between dynamic shock formation and eventual relaxation to steady-state profiles. Hence, steady-state analyses describe the final equilibrium state while dynamic analyses capture the full shock formation process. The two regimes are connected through viscous relaxation.

\section{Mathematical framework}
\label{Math_framwork}

In this paper we investigate strictly hyperbolic systems of first-order quasi-linear partial differential equations in $1+1$ dimensions of the form
\begin{equation}
\label{EOM-Matrix}
A^\alpha(\psi)\partial_\alpha \psi = A^0(\psi)\partial_t\psi + A^1(\psi)\partial_x\psi = B(\psi),
\end{equation}
where $\psi \in \Psi \subset \mathbb{R}^{N\times 1}$ is the vector containing the $N=1,2,\cdots$ unknown variables $\psi_a = \psi_1,\dots,\psi_N$, which are functions of the spacetime coordinates $(t,x)$. In particular, $\Psi$ is the space of solutions. The source term $B(\psi) \in \mathbb{R}^{N\times 1}$ and matrices $A^\alpha(\psi) \in \mathbb{R}^{N\times N}$ are assumed to be smooth functions of $\psi$. 

Let $\xi_\mu = \partial_\mu\Phi$ be the covectors normal to the characteristic hypersurface $\Sigma = \{\Phi(t,x) = 0\}$. The characteristic roots $\xi_0 = \xi_0(\xi_1)$ are solutions to the characteristic equation $\det[A^\alpha\xi_\alpha] = 0$. The system \eqref{EOM-Matrix} is
\begin{itemize}
    \item \textit{Strictly hyperbolic} if all characteristic roots $\xi_0$ are real and distinct.
    \item \textit{Causal} if $\xi_0(\xi_1)$ lies on or outside the light cone ($\xi_\mu\xi^\mu \geq 0$).
    \item \textit{Strictly causal} if $\xi_\mu\xi^\mu > 0$ for all characteristic roots.
\end{itemize}

\begin{assumption}
\label{Assumption1}
The system \eqref{EOM-Matrix} is strictly hyperbolic and strictly causal.
\end{assumption}

Under Assumption \ref{Assumption1}, choosing the timelike covector $\zeta_\mu = (1,0)$\footnote{For simplicity, we shall write $(\xi_0,\xi_1)$ instead $(\xi_0,\xi_1,\xi_2,\xi_3)$ for the $1+1$ case.} yields $\det[A^\alpha \zeta_\alpha] = \det[A^0] \neq 0$. This allows rewriting \eqref{EOM-Matrix} as
\begin{equation}
\label{EOM-Matrix2}
\partial_t \psi + a(\psi)\partial_x\psi = G(\psi),
\end{equation}
where
\begin{align*}
a(\psi) = (A^0(\psi))^{-1}A^1(\psi)\quad\text{and}\quad G(\psi) = (A^0(\psi))^{-1}B(\psi).
\end{align*}
\begin{remark}
Under Assumption~\ref{Assumption1}, $\xi_1 = 0$ is not an admissible characteristic root. 
We may therefore define $\lambda = -\xi_0/\xi_1$. 
Consequently, strict causality --- derived from the characteristic roots of the equation
\[\det[\xi_\alpha A^\alpha(\psi)] = \det[\xi_1 A^0(\psi)] \det[a(\psi) - \lambda I_N] = 0,\]
requires that the eigenvalues $\lambda$ of $a(\psi)$ lie strictly within the interval $(-1,1) \subset \mathbb{R}$. 
Furthermore, strict hyperbolicity demands that all eigenvalues $\lambda$ are simple (i.e., non-degenerate). The eigenvalues $\lambda$ are the characteristic velocities of the fluid in the laboratory frame (lab frame) and, as expected, its absolute value cannot exceed 1.
\end{remark}
For a reference state $\psi^* \in \Psi$, we impose
\begin{assumption}
\label{Assumption_T1}
The matrix $a(\psi^*)$ has only real and simple eigenvalues.
\end{assumption}
\begin{assumption}
\label{Assumption_T2}
There exists $p \in \{1,2,\dots,N\}$ such that the $p$-th eigenvalue of $a(\psi)$ is genuinely nonlinear at $\psi^*$ (defined below).
\end{assumption}
\begin{assumption}
\label{Assumption_T3}
$G(\psi^*) = 0$.
\end{assumption}
B\"arlin \cite{BARLIN2023103901} recently proved that
\begin{theorem}
\label{Theorem1}
Under Assumptions \ref{Assumption_T1}, \ref{Assumption_T2}, and \ref{Assumption_T3}, there exist smooth initial data $\overset{\circ}{\psi} \colon \mathbb{R} \to \Psi$ such that the unique $C^1$-solution of \eqref{EOM-Matrix2} with $\psi(0,x) = \overset{\circ}{\psi}(x)$ exists only for finite time, while $\psi(t,x)$ remains bounded in $\Psi$.
\end{theorem}
This establishes that smooth initial data can lead to finite-time derivative blow-ups ($\partial_x\psi \to \infty$) while the solution itself remains bounded. The proof is valid independently of the nature of the source $G$. The concept of genuinely nonlinear modes originates with Lax \cite{Lax1964}, who introduced it for systems with $N=1,2$ variables, later generalized by John \cite{John:1974} and Liu \cite{LIU197992}. Notably, Liu \cite{LIU197992} showed that for source-free systems ($G \equiv 0$), singularity formation occurs when Assumptions \ref{Assumption_T2} and \ref{Assumption_T3} hold.
%\begin{remark}
%\label{remark_III1}
%An important consequence of B\"arlin's result is that provided the initial data $\overset{\circ}{\psi}$ generates a source term $G(\overset{\circ}{\psi})$ of sufficiently small magnitude, the solution $\psi(t)$ remains in a neighborhood of the reference state $\psi^*$. Thus, the strict hyperbolicity of $\psi^*$ implies strict hyperbolicity of $\psi(t)$ for as long as the solution exists. Furthermore, if the reference state $\psi^*$ satisfies causality, and if the initial data $\overset{\circ}{\psi}$ generates a sufficiently small source term $G(\overset{\circ}{\psi}) \approx 0$, then B\"arlin's framework implies that
%\begin{enumerate}
%    \item The solution $\psi(t)$ remains close to $\psi^*$ (in an appropriate norm),
%    \item The characteristic speeds $\lambda_i(\psi(t))$ remain close to $\lambda_i(\psi^*)$.
%\end{enumerate}
%Thus, causality for $\psi^*$ \emph{typically} implies causality for $\psi(t)$, provided that the system's eigenvalues depend continuously on the state and the perturbation $G(\overset{\circ}{\psi})$ does not push $\lambda_i(\psi(t))$ beyond causal bounds.
%\end{remark}

We now clarify the concept of genuinely nonlinear modes. Let $\lambda_k(\psi)$ (for $k=1,\dots,N$) be the $N$ real and distinct eigenvalues of the matrix $a(\psi)$, with corresponding complete sets of left eigenvectors $l_i \in \mathbb{R}^{1\times N}$ and right eigenvectors $r_i \in \mathbb{R}^{N\times 1}$ satisfying
\begin{align*}
l_i(\psi)a(\psi) = \lambda_i(\psi) l_i(\psi)\quad \text{and}\quad
a(\psi)r_i(\psi) = \lambda_i(\psi)r_i(\psi).
\end{align*}
Under the normalization conditions
\begin{align}
\label{Norm}
l_i r_j = \delta_{ij} \quad \text{and} \quad |l_i| = 1,\quad \text{(Euclidean norm)}
\end{align}
we define
\begin{itemize}
\item A mode $p \in \{1,\dots,N\}$ is \textit{genuinely nonlinear} if
\begin{align}
\label{GNM}
r_p \cdot \nabla_\psi \lambda_p(\psi) \neq 0,
\end{align}
where $r_p \cdot \nabla_\psi := \sum_{a=1}^N r_p^a \frac{\partial}{\partial \psi_a}$.
\item A mode $i$ is \textit{linearly degenerate} if $r_i \cdot \nabla_\psi \lambda_i(\psi) = 0$.
\end{itemize}

To better understand these concepts, consider the spatial derivative expansion
\begin{align}
\label{omega_i}
\omega(t,x) := \partial_x \psi(t,x) = \sum_{i=1}^N \omega_i r_i, \quad \text{where} \quad \omega_i = l_i \omega.
\end{align}
The system \eqref{EOM-Matrix2} can then be expressed in characteristic form
\begin{align}
\label{Characteristic_EQ}
L_i \omega_i = \sum_{k,l=1}^N \gamma_{ikl}(\psi) \omega_k \omega_l + \sum_{k=1}^N G_{ik}(\psi) \omega_k,
\end{align}
where $L_i := \partial_t + \lambda_i \partial_x$ is the characteristic vector field and
\begin{align*}
\gamma_{ijk}&=-\frac{1}{2}\left (\delta_{ij}r_k\cdot \nabla_\psi \lambda_i+\delta_{ik}r_j\cdot \nabla_\psi \lambda_i\right )
+\frac{1}{2}(\lambda_j-\lambda_k)\left (r_k\cdot\nabla_\psi l_i \;r_j-r_j\cdot\nabla_\psi l_i \; r_k\right ),\\
G_{ik}&=G\cdot\nabla_\psi l_i\; r_k+l_i\; r_k\cdot\nabla_\psi G.
\end{align*}
We used that $G\cdot\nabla_\psi :=\sum_{a=1}^N G^a\frac{\partial}{\partial \psi^a}$. The characteristic curves $X_i(t,x)$ are solutions to
\begin{align}
\label{Characteristic_Curve}
\begin{cases}
\frac{dX_i(t,x)}{dt} = \lambda_i(X_i(t,x)), \\
X_i(0,x) = x.
\end{cases}
\end{align}
Along these curves, \eqref{Characteristic_EQ} becomes
\begin{align}
\label{Characteristic_EQ2}
\frac{d}{dt} \omega_i(t,X_i(t,x)) = \left [\sum_{k,l=1}^N \gamma_{ikl}(\psi) \omega_k \omega_l + \sum_{k=1}^N G_{ik}(\psi) \omega_k\right ]_{(t,x)=(t,X_i(t,x))}.
\end{align}
A direct consequence is that the nonlinear terms are governed by $\gamma_{ikl}$ (see \cite{John:1974} for explicit forms). The linear terms (involving $G_{ik}$) cannot generate nonlinear waves \cite{BARLIN2023103901}, and linear degenerate fields $\omega_i$ are the responsible for contact discontinuities in the ideal fluid for discontinuous initial data (Riemann problem). Genuine nonlinearity requires $\gamma_{iii} \neq 0$. The crucial coefficients are
\bml
\label{gamma_ikk}
\begin{align}
\gamma_{ikk} &= 0 \quad \text{for} \quad i \neq k, \\
\gamma_{iii} &=  -r_i \cdot \nabla_\psi \lambda_i(\psi).
\end{align}
\eml
Thus, linearly degenerate fields ($\gamma_{iii} = 0$) behave quasi-linearly and resist singularity formation, while genuinely nonlinear fields ($\gamma_{iii} \neq 0$) permit $\omega_i^2$ terms that drive wave steepening.

As shown by B\"arlin \cite{BARLIN2023103901}, for sufficiently small initial source $G$, smooth initial data exist that lead to finite-time gradient singularities. This behavior is essential for discontinuous shock formation, as it provides the mechanism for derivative blow-ups while maintaining bounded solutions.
 
\subsection*{Connection to Israel-Stewart Equations}

As discussed in Section~\ref{Sec:Discussion}, the source terms in the system are contained in $B$, which for the shear tensor relaxation equation takes the specific form $-\pi^{\mu\nu}/\tau_\pi$. This structure directly relates to B\"arlin's theorem \cite{BARLIN2023103901} through the following mechanism.
\begin{itemize}
    \item \textbf{Initial Condition Requirement}: Source terms such as $-\pi^{\mu\nu}/\tau_\pi$ must be initially small. This condition delays the relaxation of $\pi^{\mu\nu}$ toward its Navier-Stokes value $-2\eta\sigma^{\mu\nu}$ and allows nonlinear wave modes to dominate the early-time evolution.   
    \item \textbf{Singularity Formation}: The temporary suppression of damping enables gradient steepening through nonlinear effects ($u\partial_x u$ terms) and the finite-time blow-up in derivatives before viscous dissipation becomes significant.
\end{itemize}
This matches precisely the physical scenario where nonlinear effects dominate ($|u\partial_x u| \gg |\beta\partial_x^2 u|$) at early times while dissipation balances nonlinearity (Olson-Hiscock regime \cite{OlsonHiscock1990}) at late times. 

\subsection{Approximate Estimates of Gradient Blow-ups}
\label{Sec:Estimates}

To develop intuitive estimates for when gradient blow-ups may occur, we make several simplifying approximations.
\begin{itemize}
    \item Neglect source terms: $G_{ik} \approx 0$.
    \item Assume nonlinear terms dominate linear contributions in \eqref{Characteristic_EQ2} along the characteristic curve $X_i(t,x)$.
    \item Assume that $\psi\approx \psi^*$.
\end{itemize}
Under these assumptions, the evolution of a genuinely nonlinear mode $\lambda_i$ reduces to
\begin{align}
\label{nonlinear_mode_Eq}
\frac{d}{dt}\omega^*_i(t,X_i(t,x)) \approx \gamma_{iii}(\psi^*)\omega^*_i{}^2(t,X_i(t,x))
\end{align}
with approximate solution
\begin{align*}
\omega^*_i(t,X_i(t,x)) \approx \frac{\overset{\circ}{\omega}^*_i(x)}{1 - \overset{\circ}{\omega}^*_i(x)\int_0^t \gamma_{iii}(\psi^*(s,X_i(s,x)))ds}.
\end{align*}
\textbf{Blow-up Condition}: When $\overset{\circ}{\omega}^*_i(x)\gamma_{iii}(\psi^*(t,X_i(t,x))) > 0$ for $t \in [0,t_c]$, a singularity forms at time $t_c$ satisfying
\begin{align*}
\int_0^{t_c} \gamma_{iii}(\psi^*(s,X_i(s,x_0)))ds = \frac{1}{\overset{\circ}{\omega}^*_i(x_0)},
\end{align*}
where $x_0$ is the maximum point $x_0 = \max_x \overset{\circ}{\omega}^*_i(x)$ if $\gamma^*_{iii} > 0$ and the minimum point $x_0 = \min_x \overset{\circ}{\omega}^*_i(x)$ if $\gamma^*_{iii} < 0$. 

The blow-up location propagates along the characteristic curve $x_{\text{blow-up}} = x_0 + \int_0^{t_c} \lambda_i(X_i(s,x_0))ds$.\\
\textbf{Simplified Timescale Estimate}: For approximately constant $\gamma^*_{iii}$, we obtain
\begin{align}
\label{t_c}
t_c \approx \frac{1}{\overset{\circ}{\gamma}^*_{iii}(x_0) \overset{\circ}{\omega}^*_i(x_0)}.
\end{align}
This shows that larger nonlinear coupling $\gamma^*_{iii}$ leads to earlier blow-up, stronger initial gradients $\overset{\circ}{\omega}^*_i(x_0)$ accelerate singularity formation, and
that the blow-up time is inversely proportional to the product of initial gradient strength and nonlinearity.

\begin{remark}
The blow-up time estimate in Eq.~(11) provides key qualitative insight but has limited quantitative accuracy. It is derived under the approximations of a frozen coefficient matrix \( a(\psi) \approx a(\psi^*) \), negligible source terms \( G \approx 0 \), and isolated mode evolution. In practice, the nonlinear coupling between modes and the evolution of the background state \( \psi \) cause deviations from this simple estimate. Thus, Eq.~(11) reliably indicates \textit{whether} a blow-up will occur and the \textit{relative order} of blow-up times for different modes, but the precise value of \( t_c \) is not expected to match numerical results closely.
\end{remark}

For convenience, Table~\ref{table:conditions} collects the conditions for strict hyperbolicity and nonlinear causality for the three separate regimes analyzed in the following sections.
\begin{table}[h!]
\centering
\caption{Summary of causality and hyperbolicity conditions for the three Israel-Stewart regimes in $1+1$ dimensions. The conditions must hold for the system to be strictly hyperbolic and causal.}
\label{table:conditions}
\begin{tabular}{p{2.5cm} p{4cm} p{4cm} p{2cm}}
\hline
\textbf{Regime} & \textbf{Characteristic Velocity (Rest Frame)} & \textbf{Primary Conditions} & \textbf{Ref.} \\
\hline
Bulk Viscosity & $c_{\Pi} = \sqrt{P_e + \dfrac{\zeta}{\tau_{\Pi}\rho_\Pi}}$ & \begin{tabular}{l} $c_{\Pi} \in (0,1)$\\ $\rho_{\Pi} = \varepsilon+P+\Pi > 0$ \end{tabular} & Sec. IV, Assump. 5 \\
\hline
Shear Viscosity & $c_{\pi} = \sqrt{P_e + \dfrac{4\tilde{\eta}}{3\rho_{\pi}\tau_{\pi}}}$ & \begin{tabular}{l} $c_{\pi} \in (0,1)$ \\ $\rho_{\pi} = \varepsilon+P+\dfrac{\pi^{11}}{u_0^2} > 0$ \end{tabular} & Sec. V, Assump. 7 \\
\hline
Diffusion & $c_1 = \dfrac{1}{\sqrt{3}}, \quad c_2 = \sqrt{\dfrac{4\varepsilon \kappa T}{3n^2\tau_q}}$ & \begin{tabular}{l} $c_1, c_2 \in (0,1)$\\ $\quad c_1 \neq c_2$\\ $n, \tau_q, \varepsilon, T > 0$ \end{tabular} & Sec. VI, Assump. 9 \\
\hline
\end{tabular}
\end{table}

\section{Barotropic Fluid with Bulk Viscosity}
\label{Bulk}

We consider a barotropic fluid with bulk viscous corrections, where the pressure $P$ obeys an equation of state $P = P(\varepsilon)$. The energy-momentum tensor is given by
\be
\label{T}
T^{\mu\nu} = \varepsilon u^\mu u^\nu + (P + \Pi)\Delta^{\mu\nu},
\ee
where $\Pi$ represents the bulk viscous pressure. The relaxation equation for $\Pi$ takes the general form
\be
\label{Pi}
\tau_\Pi u^\alpha\partial_\alpha \Pi + \Pi = -\zeta\partial_\alpha u^\alpha,
\ee
with $\zeta = \zeta(\varepsilon,\Pi)$ being the bulk viscosity coefficient and $\tau_\Pi = \tau_\Pi(\varepsilon,\Pi)$ the relaxation time. This form is general since additional terms can be incorporated through appropriate rescaling of transport parameters, as mentioned in \cite{Gavassino:2023xkt}.

For the $1+1$ dimensional case, the complete system consists of \eqref{Pi} combined with energy-momentum conservation $\partial_\nu T^{\mu\nu} = 0$, yielding
\bml
\label{Bulk_EOM}
\begin{align}
u^\alpha\partial_\alpha\varepsilon + \rho_\Pi C^\alpha\partial_\alpha u &= 0, \\
P_e u_0^2 C^\alpha\partial_\alpha\varepsilon + \rho_\Pi u^\alpha\partial_\alpha u + u_0^2 C^\alpha\partial_\alpha\Pi &= 0, \\
\frac{\zeta}{\tau_\Pi}C^\alpha \partial_\alpha u + u^\alpha\partial_\alpha \Pi &= -\frac{\Pi}{\tau_\Pi},
\end{align}
\eml
where we define
\begin{align}
\label{C}
C^\alpha := \delta^\alpha_1 + \frac{u}{u^0}\delta^\alpha_0, \quad
\rho_\Pi := \varepsilon + P + \Pi, \quad \text{and}\quad
P_e:=\frac{\partial P}{\partial \varepsilon}.
\end{align}
The system can be expressed in matrix form \eqref{EOM-Matrix} with $\psi = (\varepsilon, u, \Pi)$, $B = (0, 0, -\Pi/\tau_\Pi)$, both belonging to $\mathbb{R}^{3\times 1}$, and
\be
A^\alpha = \bbm
u^\alpha & \rho_\Pi C^\alpha & 0 \\
P_e u_0^2 C^\alpha & \rho_\Pi u^\alpha & u_0^2 C^\alpha \\
0 & \frac{\zeta}{\tau_\Pi} C^\alpha & u^\alpha
\ebm. \quad (\alpha = 0,1)
\ee
The characteristic determinant evaluates to
\begin{align*}
\det[A^\alpha\xi_\alpha] = \rho_\Pi(u^\alpha\xi_\alpha)\bbm
(u^\alpha\xi_\alpha)^2 - c_\Pi^2\Delta^{\alpha\beta}\xi_\alpha\xi_\beta
\ebm,
\end{align*}
where the characteristic velocity squared in the fluids rest frame is \cite{Bemfica:2019cop}
\begin{align*}
c_\Pi^2 = P_e + \frac{\zeta}{\tau_\Pi(\varepsilon + P + \Pi)}.
\end{align*}
\begin{assumption}
\label{A5}
Assume $c_\Pi \in (0,1) \subset \mathbb{R}$ and $\rho_\Pi > 0$.
\end{assumption}
Under this assumption, the system is strictly hyperbolic and causal, with characteristic velocities $0$ and $\pm c_\Pi$ in the local rest frame.
\begin{remark}
As discussed in Section~\ref{Sec:Discussion}, the characteristic polynomial factors into $(u^\alpha\xi_\alpha)^2 - c_\Pi^2\Delta^{\alpha\beta}\xi_\alpha\xi_\beta$ for non-zero velocities. For $\xi_\mu = (0,1)$, we find $\det[A^1] = \rho_\Pi \gamma^3 v(v^2 - c_\Pi^2)$, showing that steady-state solutions cannot have $v$ crossing $c_\Pi$.
\end{remark}
Following Section~\ref{Math_framwork}, the invertibility of $A^0$ allows us to write \eqref{EOM-Matrix} in the form \eqref{EOM-Matrix2}. The eigenvalues of $a(\psi) = (A^0(\psi))^{-1}A^1(\psi)$ are
\be
\label{Eigenvalues}
\lambda_1 = \frac{u}{u^0} \quad \text{and} \quad
\lambda_{\pm} = \frac{u \pm u^0 c_\Pi}{u^0 \pm u c_\Pi},
\ee
which correspond to the characteristic velocities in the lab frame.
Under Assumption~\ref{A5}, these eigenvalues are real and distinct. The corresponding right eigenvectors are
\[
r_1 = \bbm -1 \\ 0 \\ P_e \ebm \quad \text{and} \quad
r_{\pm} = \bbm \pm\rho_\Pi \tau_\Pi \\ u^0 c_\Pi \tau_\Pi \\ \pm\zeta \ebm.
\]
While the mode $\lambda_1$ is linearly degenerate ($r_1 \cdot \nabla_\psi \lambda_1(\psi) = 0$), the modes $\lambda_\pm$ are genuinely nonlinear degenerate under Assumption \ref{A5} and depending on the parameters:
\begin{align*}
r_\pm\cdot\nabla_\psi \lambda_\pm=\frac{1}{(u^0\pm u c_\Pi)^2}\left [\tau_\Pi c_\Pi(1-c_\Pi^2)+\tau_\Pi\rho_\Pi\frac{\partial c_\Pi}{\partial\varepsilon}
+\zeta\frac{\partial c_\Pi}{\partial\Pi}\right ].
\end{align*}
Consider the reference state $\psi^* = (\varepsilon, u, 0)$ where $\Pi^* = 0$, satisfying $G(\psi^*) = 0$ from Assumption~\ref{Assumption_T3}. In this state:
\be
\lambda^*_1 = \lambda_1(\psi^*) = \frac{u}{u^0} \quad \text{and} \quad
\lambda^*_{\pm} := \lambda_{\pm}(\psi^*) = \frac{u \pm u^0 c^*_\Pi}{u^0 \pm u c^*_\Pi},
\ee
with $c^*_\Pi$ evaluated at $\rho_\Pi^* = \varepsilon + P$, $\zeta^* = \zeta(\varepsilon,0)$, and $\tau_\Pi^* = \tau_\Pi(\varepsilon,0)$. The eigenvectors become
\[
r_1(\psi^*) = \bbm -1 \\ 0 \\ P_e \ebm \quad \text{and} \quad
r_{\pm}(\psi^*) = \bbm \pm(\varepsilon + P) \tau^*_\Pi \\ u^0 c^*_\Pi \tau^*_\Pi \\ \pm\zeta^* \ebm.
\]
\begin{assumption}
\label{A6} 
$c_\Pi^* \in (0,1) \subset \mathbb{R}$ and $\varepsilon + P > 0$.
\end{assumption}
\begin{theorem}
\label{Theorem2}
Under Assumption~\ref{A6}, there exists smooth initial data $\overset{\circ}{\psi} \in \Psi$ such that if
\begin{align}
\label{Condi_Pi}
(\varepsilon + P)\frac{\partial c_\Pi^*}{\partial\varepsilon} + c_\Pi^*\left(1 - c_\Pi^{*2}\right) \ne 0,
\end{align}
the corresponding $C^1$ solution $\psi$ remains bounded while its gradient $\partial_x\psi$ develops finite-time blow-up.
\end{theorem}
\begin{proof}
Condition \eqref{Condi_Pi} ensures $r_\pm^* \cdot \nabla_{\psi^*} \lambda_\pm(\psi^*) \ne 0$, satisfying Assumption~\ref{Assumption_T2}. Assumption~\ref{A6} guarantees distinct eigenvalues, and $\Pi^* = 0$ satisfies Assumption~\ref{Assumption_T3}. The result follows from Theorem~\ref{Theorem1}.
\end{proof}
\begin{remark}
\label{Remark1}
The two genuinely nonlinear modes can produce blow-ups propagating in opposite or in same directions with velocities $\lambda_+$ and $\lambda_-$, while the linearly degenerate mode $\lambda_1 = v$ generates only linear waves without blow-ups.
\end{remark}

\section{Barotropic fluid with Shear Viscosity}
\label{Sec:Shear}

We consider a barotropic fluid ($P = P(\varepsilon)$) with shear viscous corrections described by the energy-momentum tensor
\[T^{\mu\nu} = \varepsilon u^\mu u^\nu + P\Delta^{\mu\nu} + \pi^{\mu\nu},\]
where the shear tensor $\pi^{\mu\nu}$ satisfies the constraints $\pi^\mu_\mu = u^\mu\pi_{\mu\nu} = 0$ and $\pi_{\mu\nu} = \pi_{\nu\mu}$. The relaxation equation is
\begin{align}
\label{Relax_Shear}
\Delta^{\alpha\beta}_{\mu\nu}\left(\tau_\pi u^\lambda\partial_\lambda \pi_{\alpha\beta} + \frac{4\delta_{\pi\pi}}{3}\pi_{\alpha\beta}\partial_\lambda u^\lambda\right) + \pi_{\mu\nu} = -2\eta\sigma_{\mu\nu},
\end{align}
where $\Delta_{\mu\nu}^{\alpha\beta} = \frac{1}{2}\left[\Delta_\mu^\alpha\Delta_\nu^\beta + \Delta_\mu^\beta\Delta_\nu^\alpha - \frac{2}{3}\Delta_{\mu\nu}\Delta^{\alpha\beta}\right]$ is the shear projector, $\sigma_{\mu\nu} = \Delta_{\mu\nu}^{\alpha\beta}\partial_\alpha u_\beta$ is the shear tensor, $\eta = \eta(\varepsilon,\pi^{\mu\nu})$ is the shear viscosity coefficient, $\tau_\pi = \tau_\pi(\varepsilon,\pi^{\mu\nu})$ is the relaxation time, and $\delta_{\pi\pi}$ is a transport parameter (equal to $\tau_\pi$ in the conformal limit).

In $1+1$ dimensions, the non-zero components of $\sigma_{\mu\nu}$ are $\sigma_{00} = \frac{u}{u_0}\sigma_{01} = \frac{u^2}{u_0^2}\sigma_{11} = \frac{2}{3}u^2 C^\alpha\partial_\alpha u$ and $\sigma_{22} = \sigma_{33} = -\frac{1}{3}C^\alpha\partial_\alpha u$. The complete PDE system becomes
\bml
\label{EOM_Shear}
\begin{align}
u^\alpha\partial_\alpha\varepsilon + \left[(\varepsilon+P) + \frac{\pi^{11}}{u_0^2}\right]C^\alpha \partial_\alpha u &= 0, \\
u_0^2 P_e C^\alpha\partial_\alpha\varepsilon + \left[(\varepsilon+P)u^\alpha + \frac{\pi^{11}}{u_0^2}\left(\frac{\delta^\alpha_0}{u^0} - uC^\alpha\right)\right]\partial_\alpha u 
+ C^\alpha\partial_\alpha\pi^{11} &= 0, \\
\left(\frac{4u_0^2\tilde{\eta}C^\alpha}{3\tau_\pi} - \frac{2 u u^\alpha \pi^{11}}{u_0^2}\right)\partial_\alpha u + u^\alpha\partial_\alpha \pi^{11} &= -\frac{\pi^{11}}{\tau_\pi},
\end{align}
\eml  
where $\tilde{\eta} = \eta + \delta_{\pi\pi}\pi^{11}/u_0^2$.

In matrix form \eqref{EOM-Matrix}, we define $\psi = (\varepsilon, u, \pi^{11})\in\mathbb{R}^{3\times 1}$, the source $B = (0, 0, -\pi^{11}/\tau_\pi)\in\mathbb{R}^{3\times 1}$, and the matrices
\be
\label{Matrix_Shear}
A^\alpha = \bbm
u^\alpha & \left[(\varepsilon+P) + \frac{\pi^{11}}{u_0^2}\right]C^\alpha & 0 \\
u_0^2 P_e C^\alpha & (\varepsilon+P)u^\alpha + \frac{\pi^{11}}{u_0^2}\left(\frac{\delta^\alpha_0}{u^0} - uC^\alpha\right) & C^\alpha \\
0 & \frac{4u_0^2\tilde{\eta} C^\alpha}{3\tau_\pi} - \frac{2 u \pi^{11} u^\alpha}{u_0^2} & u^\alpha 
\ebm.
\ee
The characteristic determinant yields
\[
\det(\xi_\alpha A^\alpha) = \rho_\pi (u^\alpha\xi_\alpha)\left[(u^\alpha\xi_\alpha)^2 - c_\pi^2 \Delta^{\mu\nu}\xi_\mu\xi_\nu\right],
\]
where~\footnote{While $\pi^{11}$ and $u^0$ are not scalars, the ratio $\pi^{11}/u_0^2$ is. Therefore, $\rho_\pi$ and $c_\pi$ are also scalars.}
\begin{align}
\label{Accoustic_vel_Shear}
\rho_\pi = \varepsilon + P + \frac{\pi^{11}}{u_0^2} \quad \text{and}\quad
c_\pi^2 = P_e + \frac{4\tilde{\eta}}{3\rho_\pi \tau_\pi}.
\end{align}
\begin{assumption}
\label{AS1}
Assume $c_\pi \in (0,1) \subset \mathbb{R}$ and $\rho_\pi > 0$.
\end{assumption}
Under Assumption \ref{AS1}, the system is strictly hyperbolic and causal. The distinct eigenvalues of $a(\psi) = (A^0)^{-1}A^1$ are
\[\lambda_1 = \frac{u}{u^0} \quad \text{and} \quad \lambda_{\pm} = \frac{u \pm c_\pi u^0}{u^0 \pm c_\pi u},\]
with right eigenvectors
\[r_1 = \bbm -1 \\ 0 \\ P_e u_0^2 \ebm \quad \text{and} \quad 
r_{\pm} = \bbm \pm\rho_\pi \\ u^0 c_\pi \\ \pm \frac{4u_0^2\tilde{\eta}}{3 \tau_\pi} + \frac{2 uc_\pi \pi^{11}}{u^0} \ebm.\]
Again, $\lambda_1$ is linearly degenerate ($r_1\cdot\nabla_{\psi}\lambda_1=0$), while $\lambda_\pm$ are genuinely nonlinear under Assumption \ref{AS1} and depending on the parameters:
\begin{align*}
r_\pm\cdot\nabla_\psi \lambda_\pm=\frac{1}{(u^0\pm uc_\pi)^2}\left [c_\pi(1-c_\pi^2)+\rho_\pi\frac{\partial c_\pi}{\partial\varepsilon}
+\left (\frac{4\tilde{\eta}u_0^2}{3\tau_\pi}\pm\frac{2u c_\pi \pi^{11}}{u^0}\right )\frac{\partial c_\pi}{\partial \pi^{11}}\right ].
\end{align*} 

\subsection*{Singularity Formation}
Consider the reference state $\psi^* = (\varepsilon, u, 0)$ with $\tau^*_\pi = \tau_\pi(\varepsilon,0)$, $\rho^*_\pi = \varepsilon + P$, $\tilde{\eta}^* = \eta^* = \eta(\varepsilon,0)$, and $c_\pi^* = P_e + \frac{4\eta^*}{3\rho \tau^*_\pi u_0^2}$. The eigenvectors simplify to
\[r_1^* = \bbm -1 \\ 0 \\ P_e u_0^2 \ebm \quad \text{and} \quad 
r_{\pm}^* = \bbm \pm(\varepsilon+P) \\ u^0 c^*_\pi \\ \pm \frac{4\eta^* u_0^2}{3\tau_\pi^*} \ebm.\]
\begin{assumption}
\label{AS2}
Assume $c_\pi^* \in (0,1)$ and $\varepsilon + P > 0$.
\end{assumption}
\begin{theorem}
\label{Theorem3}
Under Assumption \ref{AS2}, there exists smooth initial data $\overset{\circ}{\psi} \in \Psi$ such that if
\begin{align}
\label{Shear_Condition}
(\varepsilon + P)\frac{\partial c_\pi^*}{\partial \varepsilon} + c_\pi^*(1 - c_\pi^{*2}) \ne 0,
\end{align}
then the solution $\psi$ remains bounded while its gradient $\partial_x\psi$ develops finite-time blow-up.
\end{theorem}
\begin{proof}
The reference state satisfies $\pi^*{}^{\mu\nu} = 0$ (Assumption \ref{Assumption_T3}), while Assumption \ref{AS2} ensures distinct eigenvalues (Assumption \ref{Assumption_T1}). Condition \eqref{Shear_Condition} guarantees $r^*_\pm \cdot \nabla_\psi \lambda_\pm(\psi^*) \ne 0$ (Assumption \ref{Assumption_T2}), so Theorem \ref{Theorem1} applies.
\end{proof}

\begin{remark}
As in the bulk viscous case, the two genuinely nonlinear modes can produce blow-ups propagating either in opposite or same directions with velocities $\lambda_+$ and $\lambda_-$, while the linearly degenerate mode $\lambda_1$ generates only non-singular linear waves.
\end{remark}

\section{Israel-Stewart Theory with Diffusion}
\label{Sec:Diff}

We now consider the Israel-Stewart equations with only diffusion viscosity corrections in the Landau hydrodynamic frame. The system consists of an ideal energy-momentum tensor $T^{\mu\nu} = \varepsilon u^\mu u^\nu + P\Delta^{\mu\nu}$, a modified number current $J^\mu = n u^\mu + q^\mu$ where $n = N/V$ is the number density, and the conservation laws $\partial_\mu J^\mu = 0$ and $\partial_\nu T^{\mu\nu} = 0$. The diffusion current $q^\mu$ obeys the relaxation equation
\begin{align}
\label{Diff}
\tau_q \Delta^{\mu}_\nu u^\alpha\partial_\alpha q^\nu + q^\mu = -\kappa T^2\Delta^{\mu\nu}\partial_\nu\left(\frac{\mu}{T}\right),
\end{align}
where $\tau_q$ and $\kappa$ are transport coefficients that may depend on system variables but not on their derivatives, and $\mu$ is the chemical potential. Using thermodynamic relations and conservation laws, we rewrite this as
\begin{align}
\label{Diff_new}
\tau_q \Delta^{\mu}_\nu u^\alpha\partial_\alpha q^\nu + q^\mu = \frac{(\varepsilon+P)\kappa T}{n}\Delta^{\mu\nu}\left(\frac{\partial_\nu T}{T} + u^\alpha\partial_\alpha u_\nu\right).
\end{align}

For an ideal gas equation of state $P = \varepsilon/3 = nT$, and in $1+1$ dimensions, the diffusion vector simplifies to $q^\mu = (q^0, q, 0, 0)$ with $q^0 = uq/u^0$ due to orthogonality with $u^\mu$. The complete system becomes
\bml
\label{EOM_diff}
\begin{align}
u^\alpha\partial_\alpha\varepsilon + \frac{4\varepsilon}{3} C^\alpha\partial_\alpha u &= 0, \\
u_0^2 C^\alpha\partial_\alpha\varepsilon + 4\varepsilon u^\alpha\partial_\alpha u &= 0, \\
\frac{n}{\varepsilon}u^\alpha\partial_\alpha \varepsilon - \frac{n}{T}u^\alpha\partial_\alpha T + \left(n C^\alpha + \frac{q\delta^\alpha_0}{(u^0)^3}\right)\partial_\alpha u + C^\alpha\partial_\alpha q &= 0, \\
\frac{u_0^2 4\varepsilon\kappa}{3n\tau_q} C^\alpha\partial_\alpha T + \left(\frac{u q}{u_0^2} + \frac{4\varepsilon\kappa T}{3n\tau_q}\right) u^\alpha\partial_\alpha u - u^\alpha\partial_\alpha q &= \frac{q}{\tau_q}.
\end{align}
\eml
In matrix form \eqref{EOM-Matrix}, we define $\psi = (\varepsilon, T, u, q)\in\mathbb{R}^{4\times 1}$, $B = (0, 0, 0, q/\tau_q)\in\mathbb{R}^{4\times 1}$, and the matrices
\be
A^\alpha = \bbm
u^\alpha & 0 & \frac{4\varepsilon}{3} C^\alpha & 0 \\
u_0^2 C^\alpha & 0 & 4\varepsilon u^\alpha & 0 \\
\frac{n}{\varepsilon}u^\alpha & -\frac{n}{T}u^\alpha & \left(n C^\alpha + \frac{q\delta^\alpha_0}{(u^0)^3}\right) & C^\alpha \\
0 & \frac{u_0^2 n c_2^2}{T} C^\alpha & \left(\frac{u q}{u_0^2} - n c_2^2\right) u^\alpha & -u^\alpha
\ebm,
\ee
where we introduce the characteristic velocities
\[
c_1 = \frac{1}{\sqrt{3}} \quad \text{and} \quad c_2 = \sqrt{\frac{4\varepsilon\kappa T}{3n^2\tau_q}}.
\]

The characteristic determinant reveals
\[
\det(\xi_\alpha A^\alpha) = -4\varepsilon\frac{n \tau_q}{T}\left[(u^\alpha\xi_\alpha)^2 - c_1^2\Delta^{\mu\nu}\xi_\mu\xi_\nu\right]\left[(u^\alpha\xi_\alpha)^2 - c_2^2 \Delta^{\mu\nu}\xi_\mu\xi_\nu\right].
\]
Causality and strictly hyperbolicity hold under the following assumption:
\begin{assumption}
\label{AD1}
Assume $n, \tau_q, \varepsilon, T > 0$, with
\begin{align}
\label{c_2}
c_2 \neq c_1 \quad \text{and} \quad c_2 \in (0,1).
\end{align}
\end{assumption}
Under such assumption, $A^0$ is invertible and we may rewrite the system as in \eqref{EOM-Matrix2} after defining the matrix $a(\Psi)=(A^0)^{-1}A^1$. The eigenvalues of $a(\psi)$ are
\[
\lambda_{1\pm} = \frac{u \pm u^0 c_1}{u^0 \pm u c_1} \quad \text{and} \quad \lambda_{2\pm} = \frac{u \pm u^0 c_2}{u^0 \pm u c_2},
\]
with the respective right eigenvectors
\[r_{1\pm} = \bbm
\pm \frac{4c_1\varepsilon}{u^0} \\
-\frac{q T c_1^2}{n u_0^2(c_2^2-c_1^2)} \pm \frac{T c_1}{u^0} \\
1 \\
\pm \frac{q c_1c_2^2}{u^0 (c_2^2-c_1^2)} + \frac{u q}{u_0^2}
\ebm, \quad
r_{2\pm} = \bbm
0 \\
\pm T \\
0 \\
- u^0 n c_2
\ebm.\]
The corresponding genuinely nonlinear modes are
\begin{align*}
r_1^\pm \cdot \nabla_\psi \lambda_1^\pm &= \frac{1-c_1^2}{u^0(u^0 \pm u c_1)^2},\\
r_2^\pm \cdot \nabla_\psi \lambda_2^\pm &= \frac{1}{(u^0 \pm u c_2)^2}\left (T\frac{\partial c_2}{\partial T}\mp n u^0\frac{\partial c_2}{\partial q}\right ).
\end{align*}

\subsection*{Singularity Formation}
For the reference state $\psi^* = (\varepsilon, T, u, 0)$ with $q^* = 0$, we define $\tau_q^* = \tau_q(\psi^*)$, $\kappa^* = \kappa(\psi^*)$, and $c_2^* = \sqrt{4\varepsilon T \kappa^*/(3n^2\tau_q^*)}$. The reference eigenvectors simplify to
\[r_{1\pm}^* = \bbm
\pm \frac{4c_1\varepsilon}{u^0} \\
\pm \frac{T c_1}{u^0} \\
1 \\
0
\ebm, \quad
r_{2\pm}^* = \bbm
0 \\
\pm T \\
0 \\
-u^0 n c_2^*
\ebm.\]
\begin{theorem}
\label{Theorem4}
Under Assumption \ref{AD1}, there exists smooth initial data $\overset{\circ}{\psi} \in \Psi$ such that the solution $\psi$ remains bounded while its gradient $\partial_x\psi$ develops finite-time blow-up.
\end{theorem}
\begin{proof}
Note that Assumptions \ref{Assumption_T1} and \ref{Assumption_T3} are automatically satisfied under Assumption \ref{AD1} for $q^*=0$. Moreover,
the nonlinearity conditions become
\[
r_1^{*\pm} \cdot \nabla_\psi \lambda_1^{*\pm} = \frac{1-c_1^2}{u^0(u^0 \pm u c_1)^2}, \quad
r_2^{*\pm} \cdot \nabla_\psi \lambda_2^{*\pm} = \frac{1}{(u^0 \pm u c^*_2)^2}\frac{\partial c^*_2}{\partial T},
\]
which guarantees at least two genuinely nonlinear modes, in accordance with Assumption \ref{Assumption_T2}. The result follows from Theorem \ref{Theorem1}.
\end{proof}
\begin{remark}
The system exhibits four potential nonlinear modes when $\partial c_2^*/\partial T \neq 0$, which can lead to gradient blow-ups propagating with four distinct velocities $\lambda_1^\pm$ and $\lambda_2^\pm$. If $c_2^*$ does not depend on $T$, only two nonlinear modes ($\lambda_1^\pm$) remain.
\end{remark}

\begin{remark}
Finite-time gradient blow-ups are not limited to this specific choice of equation of state. However, for a general equation of state, the characteristic polynomial becomes prohibitively complex. While causality analysis remains possible, it is rendered intricate by this complexity. Moreover, the form of the characteristic equations precludes analytical expressions for the eigenvalues of $a(\psi)$. Without these eigenvalues, we cannot determine the conditions for the existence of genuinely nonlinear modes.   
\end{remark}

\section{Shock Formation in Bulk Viscous Equations - Numerical Study}
\label{Numerical_Bulk}

We numerically investigate shock formation in the Israel-Stewart equations with bulk viscosity. We have selected this regime for our numerical demonstration; the qualitative shock-forming behavior is expected to be similar in the shear and diffusion cases due to the shared mechanism of genuine nonlinearity. We study two distinct initial configurations.
\begin{itemize}
\item \textbf{Case 1}: Generates two gradient blow-ups, producing shock waves propagating in opposite directions. 
\item \textbf{Case 2}: Eliminates one nonlinear mode's blow-up (yielding a rarefaction wave) while preserving the other's shock formation.
\end{itemize}
For both cases, we adopt the following parametrization:
\begin{align}
\label{Choices_Bulk}
P_e = \frac{1}{3}, \quad \frac{\zeta}{\tau_\Pi(\varepsilon+P)} = \frac{1}{12}, \quad \text{and} \quad \zeta = (\varepsilon+P)^{3/4}.
\end{align}
This choice yields constant $P = \varepsilon/3$ and $c_\Pi^* = \sqrt{5/12} \approx 0.65$, while $c_\Pi$ remains variable.
\begin{remark}
While these parameters were chosen for simplicity, we verified that similar shock/rarefaction structures emerge for different equations of state, various $\zeta$ and $\tau_\Pi$ profiles (including constant values), and a range of initial conditions.
\end{remark}
%\begin{remark}
%The choice $c_\Pi^* = \text{constant}$ guarantees that the reference state $\psi^*$ is strictly hyperbolic and causal. By B\"arlin's result (see Remark~\ref{remark_III1}), if the source term $G(\overset{\circ}{\psi})$ is sufficiently small initially, then the solution $\psi(t)$ preserves strict hyperbolicity and causality for all $t \geq 0$ (or until blow-up, if applicable).
%\end{remark}
In order to weight the magnitude of the bulk viscosity being used, let us take the case of an ultrarelativistic gas with $\varepsilon = \kappa T^4$, where the entropy density $s = S/V$ gives the viscosity ratio
\[\frac{\zeta}{s} = \left(\frac{3}{4\kappa}\right)^{1/4}.\]
For example, typical values for the quark-gluon plasma, where actual lattice results \cite{Borsanyi:2010cj} for $\varepsilon/T$ approach the Stefan-Boltzmann limit $\kappa=(\pi^2/30)g_{\text{QGP}}$ ($g_{\text{QGP}}$ counting the number of degrees of freedom of gluons and quarks),  are $\kappa\approx 12.2\sim 15.6$, what gives $\zeta/s$ in the range $0.47 \sim 0.50 $. These values demonstrate that shocks form even with substantial dissipation ($\zeta/s$).

\subsection*{Characteristic Mode Decomposition}
Using normalized eigenvectors ($l_i r_j = \delta_{ij}$, $|l_i| = 1$)
\begin{align*}
r_1 &= \frac{N_1}{c_\Pi^2\rho_\Pi\tau_\Pi}\bbm -1 \\ 0 \\ P_e \ebm, \quad 
r_{\pm} = \frac{N_2}{2c_\Pi^2\rho_\Pi\tau_\Pi}\bbm \pm\rho_\Pi\tau_\Pi \\ u^0 c_\Pi \tau_\Pi \\ \pm \zeta \ebm, \\
l_1^T &= \frac{1}{N_1}\bbm -\zeta \\ 0 \\ \rho_\Pi\tau_\Pi \ebm, \quad 
l_\pm^T = \frac{1}{N_2}\bbm \pm P_e \\ \frac{\rho_\Pi c_\Pi}{u^0} \\ \pm 1 \ebm,
\end{align*}
where $N_1 = \sqrt{\zeta^2+\rho_\Pi^2\tau_\Pi^2}$ and $N_2 = \sqrt{P_e^2+1+\rho_\Pi^2 c_\Pi^2/u_0^2}$, the characteristic modes evolve as [see \eqref{omega_i}]
\bml
\begin{align}
\omega_1 &= \frac{1}{N_1}\left(-\zeta \partial_x \varepsilon + \frac{\rho_\Pi \tau_\Pi}{u^0}\partial_x \Pi\right), \label{omega1} \\
\omega_\pm &= \frac{1}{N_2}\left(\pm P_e\partial_x \varepsilon + \frac{\rho_\Pi c_\Pi}{u^0}\partial_x u \pm \partial_x\Pi\right). \label{omega_pm}
\end{align}
\eml
From Eq.\ \eqref{omega1}, we expect that the linear mode $\omega_1$ becomes more visible when initial profiles $\overset{\circ}{\varepsilon}(x)$ or $\overset{\circ}{\Pi}(x)$ are non-constant, though it never dominates the dynamics. For initially constant $\varepsilon$ and $\Pi$, $\omega_1$ may still emerge through mode coupling during evolution, but remains subdominant. 

All numerical simulations were performed using two high-resolution shock-capturing schemes: MUSCL (Monotonic Upstream-centered Scheme for Conservation Laws) and
WENO-Z (Weighted Essentially Non-Oscillatory, Z-version), both specifically designed for accurate shock wave simulations. Complete implementation details are provided in Section~\ref{Numerical_Scheme}. The numerical code supporting this study is available in Ref. \cite{Bemfica_code}.

\subsection{Oppositely-Propagating Shock Formation (Case 1)}
\label{subsec:case1}

We initialize a smooth velocity gradient that evolves into two counter-propagating shock waves
\begin{align}
\label{Data_Bulk}
\overset{\circ}{\varepsilon}(x) = 0.5\,\text{GeV}^4, \quad
\overset{\circ}{u}(x) = 0.5-0.5\tanh{5x}, \quad\text{and}\quad\overset{\circ}{\Pi}(x) = 0.
\end{align}
This initial condition features uniform energy density, a smooth velocity transition from $1$ to $0$ (creating opposing flows since $\lambda_-<0$ and $\lambda_+>0$), and zero initial bulk viscosity.

The system's evolution demonstrates (see Fig.\ \ref{Fig1}) gradient steepening at two distinct points due to nonlinear wave modes, corresponding to a left-moving shock (I) and a right-moving shock (III). Along the evolution, preservation of strict hyperbolicity and causality ($\varepsilon+P+\Pi > 0$ and $c_\Pi\in(0,1)$ throughout) was observed.
\begin{figure}[htb]
\centering
\includegraphics[width=8cm]{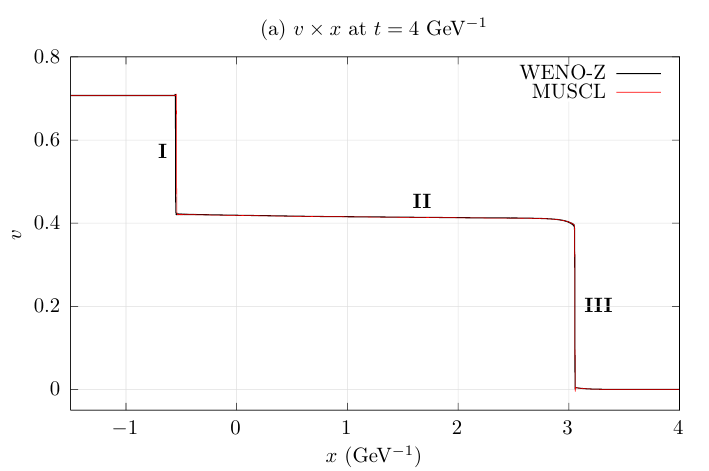}
\includegraphics[width=8cm]{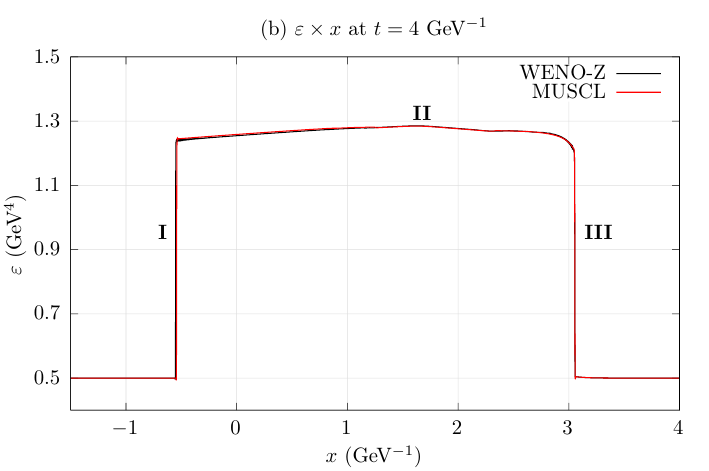}\\
\includegraphics[width=8cm]{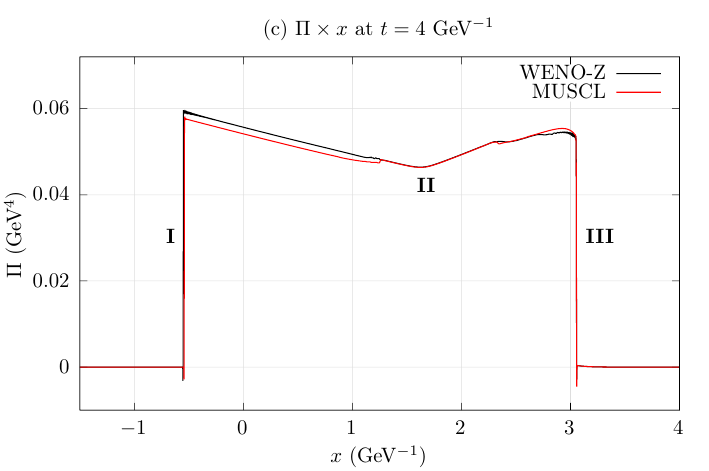}
\includegraphics[width=8cm]{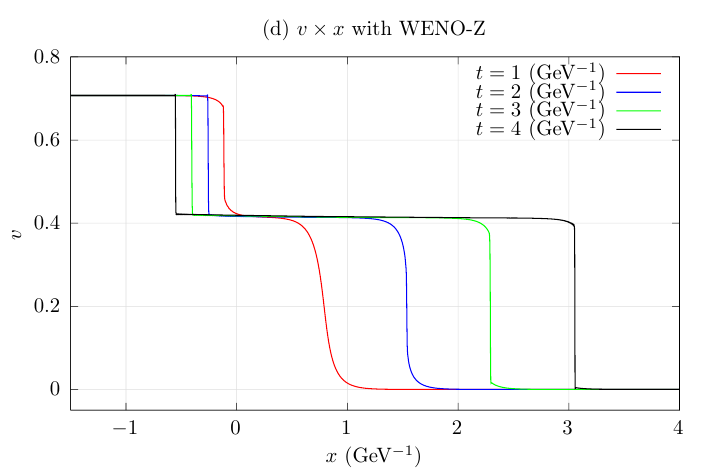}
\caption{\label{Fig1} Time evolution at $t=4\,\text{GeV}^{-1}$ showing profiles of (a) velocity $v$, (b) energy density $\varepsilon$, and (c) bulk viscosity $\Pi$ with WENO-Z and MUSCL numerical schemes for comparison. We used $\Delta x=1/400$. Figure (d) contains the velocity profile at $t=1,2,3,4\,\text{GeV}^{-1}$, solved with WENO-Z. Mode features are labeled: (I) shock ($\lambda_-$ mode), (II) linear wave, (III) shock ($\lambda_+$ mode).}
\end{figure}

For the nonlinear modes $\lambda_\pm$, we find
\bml
\label{gamma}
\begin{align}
\gamma_{+++}(\psi^*) &= -\frac{N^*_2(1-c^*_\Pi{}^2)}{2c_\Pi\rho_\Pi^*(u^0+u c^*_\pi)^2}, \\
\gamma_{---}(\psi^*) &= -\frac{N^*_2(1-c^*_\Pi{}^2)}{2c_\Pi\rho_\Pi^*(u^0-u c^*_\pi)^2},
\end{align}
\eml
both negative. The initial mode amplitudes $\overset{\circ}{\omega}^*_+=\overset{\circ}{\omega}^*_- = \overset{\circ}{l}^*_\pm \partial_x\overset{\circ}{\psi^*} = -5\rho_\Pi^*c_\Pi^*\sech^2(5x)/2u^0N^*_2$ are also negative, with $x_0$ being the minimum point. Consequently,
\begin{align*}
\overset{\circ}{\gamma}^*_{---}(x_0)\overset{\circ}{\omega}^*_-(x_0) > \overset{\circ}{\gamma}^*_{+++}(x_0)\overset{\circ}{\omega}^*_+(x_0) > 0.
\end{align*}
In accordance with the estimates in Sec.\ \ref{Sec:Estimates}, this predicts earlier blow-up in the nonlinear mode $\lambda_-$, as can be verified in Fig.\ \ref{Fig1} (d), which contains the evolution of the velocity profile, with the inset highlighting the gradient blow-ups. Figures~\ref{Fig1} (a), (b), and (c) show respectively, at $t=4\,\text{GeV}^{-1}$, the velocity profile $v$ with two shock waves, the energy density profile $\varepsilon$ showing compression jumps, and the bulk viscosity $\Pi$ developing non-zero values, also jumping to positive values when entering a shock. The modes propagate with characteristic velocities $\lambda_- < \lambda_1 < \lambda_+$, labeled respectively as (I) left-moving shock ($\lambda_-$ mode), (II) linear wave ($\lambda_1$ mode), and (III) right-moving shock ($\lambda_+$ mode). From the point of view of the shock rest frame, the upstream fronts are the left of (I) and the right of (III).

\subsection*{Shock Analysis}
\begin{description}
\item[Shock (I)] - Left-moving (computed at $t=2\,\text{GeV}^{-1}$)
\begin{itemize}
\item Shock velocity in the lab frame: $v_I \approx -0.15$.
\item Upstream/Downstream velocities in the lab frame: $v_u \approx 0.71$, $v_d \approx 0.42$.
\item Characteristic velocity at the upstream/downstream: $c_\Pi^u\approx 0.65$, $c_\Pi^d\approx 0.64$.
\item Shock frame velocities: $v'_u \approx 0.77$, $v'_d \approx 0.53$.
\item Mach numbers: $\mathcal{M}_u \approx 1.44$, $\mathcal{M}_d \approx 0.75$.
\item RH conditions satisfied with $\left\llbracket T^{\prime 0\mu}\right \rrbracket/T^{\prime 0\mu}_u <0.1\%$ for $\mu=0,1$, $\left \llbracket T^{\prime 11}\right \rrbracket \approx -5.6\times 10^{-4}$, and $\left \llbracket T^{\prime 01}\right \rrbracket \approx -8.2\times 10^{-4}$.
\end{itemize}

\item[Shock (III)] - Right-moving (computed at $t=4\,\text{GeV}^{-1}$)
\begin{itemize}
\item Shock velocity in the lab frame: $v_{III} \approx 0.77$.
\item Upstream: $v_u \approx 0$; Downstream: $v_d \approx 0.40$.
\item Characteristic velocity at the upstream/downstream: $c_\Pi^u\approx 0.66$, $c_\Pi^d\approx 0.64$.
\item Shock frame velocities: $v'_u \approx -0.77$, $v'_d \approx -0.53$.
\item Mach numbers: $\mathcal{M}_u \approx 1.41$, $\mathcal{M}_d \approx 0.75$.
\item RH conditions satisfied with $\left\llbracket T^{\prime 0\mu}\right \rrbracket/T^{\prime 0\mu}_u <0.2\%$ for $\mu=0,1$, $\left \llbracket T^{\prime 11}\right \rrbracket \approx -2.1\times 10^{-3}$, and $\left \llbracket T^{\prime 01}\right \rrbracket \approx 2.2\times 10^{-3}$.
\end{itemize}
\end{description}
Both shocks exhibit supersonic upstream flow ($\mathcal{M} > 1$), with characteristic velocity crossing and entropy density increase across shocks what, together, guarantee irreversible, physical shocks according to the Lax entropy conditions \cite{Lax1964}. All Mach numbers are computed in the shock frame (shock rest frame), as prescribed in \cite{landau1987fluid}. They correspond to the ratio $\mathcal{M}=v\gamma_v/c_\Pi\gamma_{c_\pi}$, where $\gamma_v$ and $\gamma_{c_\Pi}$ are the Lorentz factor associated with the velocities $v$ and $c_\Pi$, respectively \cite{RezzollaBook}. 

\begin{remark}
In all cases, the shock speed $v_s$ in the lab frame was estimated numerically by tracking the largest gradient in the energy density. This estimate was subsequently refined using the Rankine-Hugoniot conditions in the lab frame \cite{RezzollaBook}, which consist of the identities  
\[\left\llbracket T^{\mu\nu} \right\rrbracket n_\nu = \gamma_{v_s} \left( \left\llbracket T^{\mu 1} \right\rrbracket - v_s \left\llbracket T^{\mu 0} \right\rrbracket \right) = 0,\]  
where \(\gamma_{v_s} = 1/\sqrt{1 - v_s^2}\) is the Lorentz factor, and $n_\nu = \gamma_{v_s} (-v_s, 1, 0, 0)$ is the normal vector to the shock front, pointing in the direction of shock propagation. Here, $\left\llbracket T^{\mu\nu} \right\rrbracket = T^{\mu\nu}_d - T^{\mu\nu}_u$ denotes the jump in the stress-energy tensor across the shock. The upstream and downstream points were considered as five or more points to the right and to the left of the greatest energy gradient, symmetrically chosen.

By setting $\mu = 0$ and $\mu = 1$, the shock speed $v_s$ was determined with a precision of up to five digits when comparing the two conditions. This high level of agreement indicates that the RH conditions are in fact satisfied. 
\end{remark}

\subsection{Rarefaction-Shock Wave Configuration (Case 2)}
\label{subsec:case2}

Consider the initial data
\begin{align}
\label{Data_Bulk2}
\overset{\circ}{\varepsilon}(x) = 2.5 - 2\tanh(5x), \quad \overset{\circ}{u}(x) = 0, \quad\text{and}\quad
\overset{\circ}{\Pi}(x) = 0.
\end{align}
\begin{figure}[htb]
\centering
\includegraphics[width=8cm]{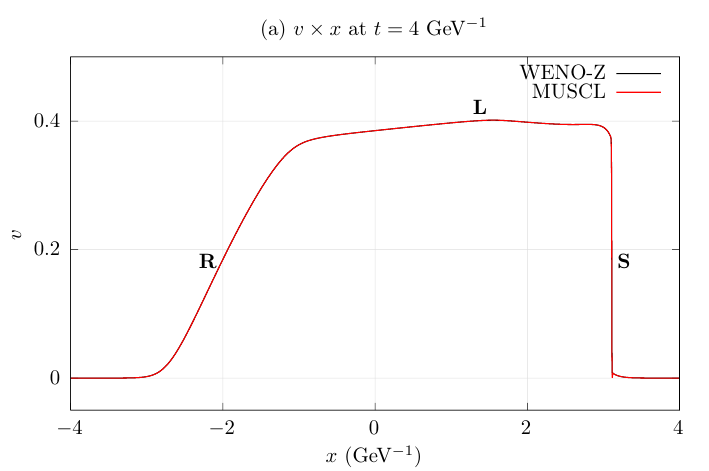}
\includegraphics[width=8cm]{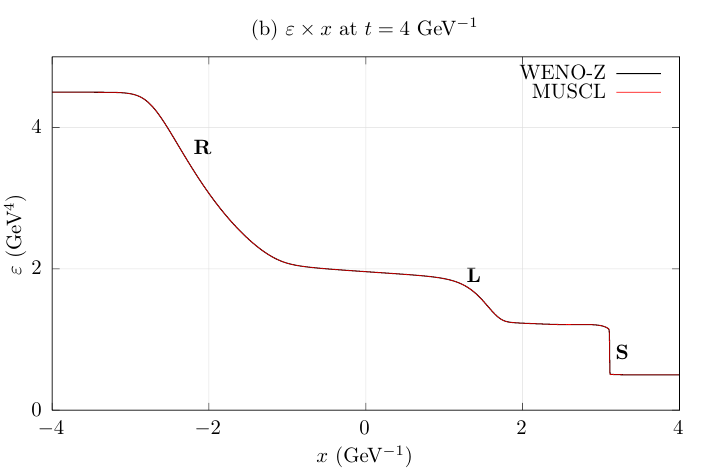}\\
\includegraphics[width=8cm]{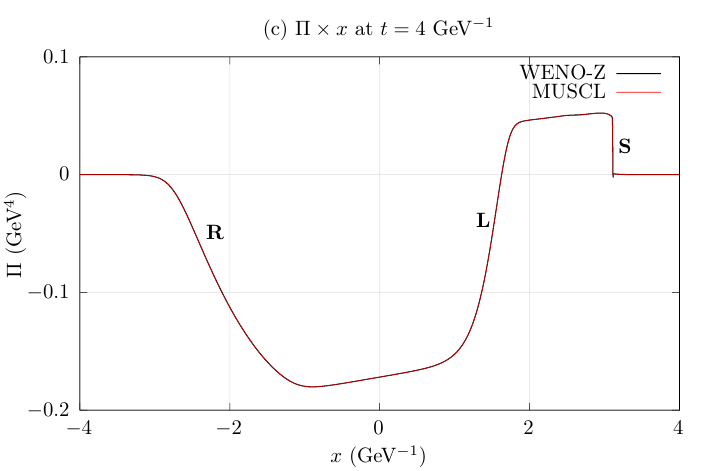}
\includegraphics[width=8cm]{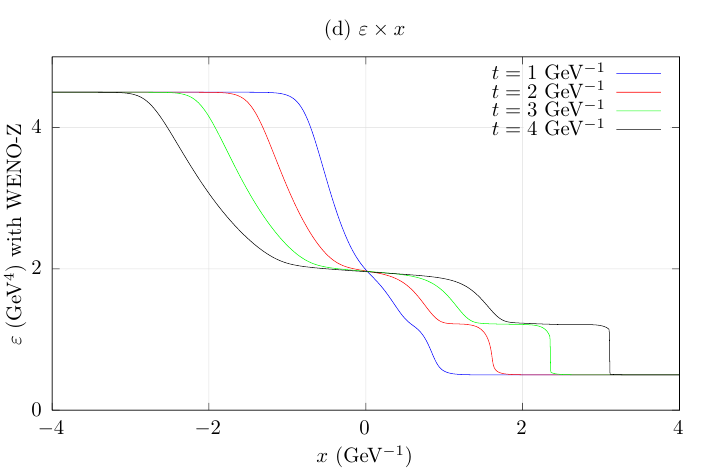}
\caption{\label{Fig2} Evolution at $t=4\,\text{GeV}^{-1}$ showing: (a) velocity $v$, (b) energy density $\varepsilon$, (c) bulk viscosity $\Pi$ (WENO-Z and MUSCL), and (d) $\varepsilon$ evolution at $t=1,2,3,4\,\text{GeV}^{-1}$ (WENO-Z). We used grid size $\Delta x =1/400$. Features labeled: (S) shock, (R) rarefaction, (L) linear wave.}
\end{figure}
The system maintains $\varepsilon+P+\Pi>0$ and $c_\Pi\in(0,1)$ (strict hyperbolicity) along the evolution. From \eqref{gamma}, both factors $\gamma_{+++}(\psi^*)$ and $\gamma_{---}(\psi^*)$ are negative for all $(x,t)$ since $c^*_\Pi<1$, while the initial mode amplitudes are $\overset{\circ}{\omega}^*_\pm = \mp 10P_e\sech^2(5x)/N^*_2$. According to the estimates in Sec.\ \ref{Sec:Estimates}, we expect no $\lambda_-$ blow-up since $\gamma_{---}^*(t,x)\overset{\circ}{\omega}^*_-(x_0)<0$ (forms rarefaction wave), while we expect a blow-up for $\lambda_+$ because $\gamma_{+++}^*(x,t)\overset{\circ}{\omega}^*_+(x_0) > 0$ (forms shock). The linear mode $\lambda_1$ creates intermediate behavior, while being more visible when compared to the case 1 (where $\omega_1=0$ initially) since it is not initially zero.

Figs. \ref{Fig2} (a)-(c) show $v$, $\varepsilon$, and $\Pi$ profiles at $t=4\,\text{GeV}^{-1}$. As can be seen, there is a shock wave (S) from $\lambda_+$ mode (right-propagating), a smooth rarefaction wave (R) from $\lambda_-$ mode (left-propagating), and a linear wave (L) from $\lambda_1$ mode. Fig.\ \ref{Fig2} (d) displays the time evolution of $\varepsilon$.
\subsection*{Shock Analysis}
\begin{itemize}
\item Shock velocity in the lab frame: $v_s \approx 0.76$. 
\item Upstream/Downstream velocities: $v_u \approx 0$, $v_d \approx 0.39$.
\item Characteristic velocity at the upstream/downstream: $c_\Pi^u\approx 0.65$, $c_\Pi^d\approx 0.64$.
\item Velocities in the shock frame: $v'_u \approx -0.76$, $v'_d \approx -0.54$.
\item Mach numbers in the lab frame: $\mathcal{M}_u \approx 1.37$, $\mathcal{M}_d \approx 0.76$.
\item RH conditions satisfied with $\left\llbracket T^{\prime 0\mu}\right \rrbracket/T^{\prime 0\mu}_u<1.2\times 10^{-2}\%$ for $\mu=0,1$, $\left\llbracket T^{\prime 11}\right\rrbracket \approx 1.3\times10^{-4}$, and $\left\llbracket T^{\prime 01}\right\rrbracket \approx -1.4\times10^{-4}$.
\end{itemize}

\begin{remark}
The numerical simulations reveal a specific regime of validity for shock formation within the Israel-Stewart framework. Successful shock formation occurs when the initial bulk viscous pressure is bounded ($|\Pi|/\tau_\Pi \lesssim 1$ in our units), ensuring the conditions for hyperbolicity and causality ($\varepsilon+P+\Pi > 0$, $c_\Pi \in (0,1)$) are maintained throughout the evolution. When $|\Pi|$ becomes too large, these fundamental conditions are violated, leading to a breakdown of the theory's mathematical well-posedness. This indicates that the bounds on $\Pi$ where violations occur are not merely a limitation of our initial data but represent the maximum departure from equilibrium that the causal Israel-Stewart theory can physically describe while maintaining its foundational properties. Consequently, there is a maximum shock strength that can be captured within this model's domain of validity.
\end{remark}

\section{Conclusions}
\label{Sec:Conclusion}

This work establishes several key theoretical and numerical results for relativistic viscous fluids. We have demonstrated, for the first time, the existence of smooth initial data that leads to finite-time gradient blow-ups in Israel-Stewart theories in $1+1$ dimensions. This occurs in three distinct regimes: pure bulk viscosity, shear viscosity, and diffusion. Through numerical simulation of the pure bulk case (without number current), we verified that the formed shocks satisfy the Rankine-Hugoniot conditions,
the characteristic velocity crossing occurs (with Mach numbers $\mathcal{M}_u > 1 > \mathcal{M}_d$), and that the upstream/downstream transitions are thermodynamically consistent. Our results complement the seminal steady-state analyses of Olson-Hiscock \cite{OlsonHiscock1990} (viscosity-dominated late-time solutions) and Geroch-Lindblom \cite{GEROCH1991394} (causality/stability constraints) by focusing on the previously unexplored \textit{early-time dynamical regime} where nonlinear effects dominate viscous damping. While Olson-Hiscock's steady-state analysis precludes shocks by construction (via time-independent solutions), our work reveals how shocks dynamically emerge from smooth data before viscosity dominates. In this context, it would be interesting to investigate whether shocks are formed in realistic hydrodynamic simulations of the quark-gluon plasma formed in heavy-ion collisions.

This analysis has been conducted with separate treatment of viscous corrections (bulk, shear, and diffusion). While future work could extend this to combined viscous effects and then to a full $3+1$ dimensional system, our present focus was to rigorously demonstrate finite-time gradient blow-ups in IS theories. These results complement the important steady-state analyses of \cite{OlsonHiscock1990} and \cite{GEROCH1991394} by revealing the early-time dynamics where nonlinear effects dominate. Their viscosity-dominated, late-time solutions --- where steady-state approximations become valid and shocks are prohibited --- emerge naturally from our dynamical framework after sufficient viscous damping occurs. The key insight is that shock formation requires analyzing the full time-dependent equations before viscosity smooths out the gradients.

In the context of relativistic heavy-ion collisions, our work suggests that shock formation could, in principle, occur in the very early pre-equilibrium stages where viscous damping is not yet dominant. However, the highly longitudinal expansion and complex 3D geometry of a heavy-ion collision may inhibit the development of the strong compression waves necessary for blow-up. The conditions we identify --- specifically, initial data with sufficiently large gradients where dissipative source terms are initially small --- are more likely to be met in regions of high initial energy density variance or in the collision of small, dense systems. Therefore, while our results establish the mathematical possibility of shocks in IS theory, their observational signature in experiments remains an open question that warrants further investigation with more realistic 3+1D simulations.

While this work establishes the phenomenon of gradient blow-up in simplified, separate regimes, extending this analysis to the full Israel-Stewart theory, where bulk, shear, and diffusion act concurrently, taking into account a more general equation of state, presents significant new challenges. A collective treatment leads to a larger system of equations with a more complex characteristic structure. The characteristic polynomial is of a higher degree, often precluding closed-form expressions for the characteristic velocities and making the identification of genuinely nonlinear modes analytically more involved. The coupling between the different viscous sectors would also complicate the evolution of characteristics and the construction of initial data guaranteed to produce a blow-up. Our results in these simpler, illustrative regimes are crucial first steps. They demonstrate unequivocally that the nonlinear hyperbolic nature of Israel-Stewart theory can dominate its dissipative character at early times, leading to shock formation --- a key physical insight that was previously missing. The methodologies developed here provide a foundation for future studies of the full theory.

Moreover, our 1+1 dimensional results do not guarantee equivalent behavior in higher dimensions, where new mathematical and numerical challenges arise. Mathematical challenges in proving blow-ups in $3+1$ dimensions, which requires fundamentally new techniques beyond our current framework. Numerical challenges because standard schemes face a critical trade-off: optimization for smooth solutions may artificially dampen shocks, while shock-capturing schemes may obscure the true discontinuity structure. Our work provides key insights into initial data structures that promote shock formation. The central message remains that shock formation in viscous relativistic fluids requires careful consideration of both the initial data structure and numerical approach, particularly when extending beyond the simplified cases studied here. In this context, it would be interesting to investigate shock formation in other causal and stable theories of relativistic fluid dynamics, such as the Bemfica-Disconzi-Noronha-Kovtun (BDNK) formalism \cite{Bemfica:2017wps,Bemfica:2019knx,Bemfica:2020zjp,Kovtun:2019hdm,Kovtun:2020eho}, where it has been shown in Ref. \cite{Freistuhler:2021lla} that, for a steady-state traveling wave solution (an ODE written in terms of the self-similar parameter $\xi=x-v_s t$, different from the time-independent Olson-Hiscock/Geroch-Lindblom analyses), shocks of sufficiently small amplitude admit smooth profile, while those with non-small amplitudes do not, with the special case of all shocks having smooth profile given that one of the characteristic velocities is the speed of light. Numerical simulations of BDNK theory have been reported in Refs. \cite{Pandya:2021ief,Pandya:2022pif,Bea:2021zol}, but a detailed study of shock formation, such as done here in the case of IS theory, has not yet been performed.

\section*{Acknowledgments} 

The author is grateful to Marcelo M. Disconzi for invaluable discussions and insights and to Jorge Noronha for carefully reading this work and contributing to its improvement. Special thanks are extended to Teerthal Patel for extensive discussions and critical contributions to the numerical aspects of this work. The manuscript was proofread and corrected with the assistance of an AI-based tool. 

\appendix

\section{Numerical scheme}
\label{Numerical_Scheme}

In our numerical simulations, we employed both the Monotonic Upstream-centered Scheme for Conservation Laws (MUSCL) \cite{VanLeer1979} with Total Variation Diminishing (TVD) and a new version of the Weighted Essentially Non-Oscillatory (WENO) scheme known as WENO-Z \cite{Borges2008}, a higher-order reconstruction method. These finite volume schemes are well-suited as shock-capturing methods, preserving fluxes across discontinuities and, consequently, the Rankine-Hugoniot conditions.  

To implement these schemes, we promoted the set $\{T^{00}, T^{01}, \Pi\}$ as the dynamical variables instead of $\{\varepsilon, u, \Pi\}$ using the transformation  
\begin{align*}  
\varepsilon &= \frac{-\Pi - T^{00}(1 - P_e) + \sqrt{[\Pi + T^{00}(1 - P_e)]^2 + 4 P_e \left[T^{00}(\Pi + T^{00}) - (T^{01})^2\right]}}{2 P_e}, \\  
u &= \sinh\left[\frac{1}{2} \arcsinh\left(\frac{2 T^{01}}{(1 + P_e)\varepsilon + \Pi}\right)\right],  
\end{align*}  
such that the energy-momentum conservation equations, together with the relaxation equation \eqref{Pi}, allow us to rewrite the equations of motion in the flux-conservative form  
\begin{align}  
\label{FluxConservativeEq}  
\partial_t q + \partial_x F = S,  
\end{align}  
where  
\[  
q = \begin{bmatrix} T^{00} \\ T^{01} \\ \Pi \end{bmatrix}, \quad  
F = \begin{bmatrix} T^{01} \\ -\varepsilon + \rho_\Pi u_0^2 \\ \frac{u}{u^0} \Pi \end{bmatrix}, \quad \text{and} \quad  
S = \begin{bmatrix} 0 \\ 0 \\ \frac{\Pi \partial_x u}{(u^0)^3} - \frac{\zeta C^\alpha \partial_\alpha u + \Pi}{u^0 \tau_\Pi} \end{bmatrix}.  
\]  

After discretizing the interval $[-L/2, L/2]$ into small uniform cells $I_i = [x_i - \Delta x/2, x_i + \Delta x/2]$ with grid size $\Delta x$, the integral of \eqref{FluxConservativeEq} over $I_i$ yields  
\begin{align*}  
\frac{d q_i}{dt} + \frac{1}{\Delta x} \left[F^*_{i+\frac{1}{2}} - F^*_{i-\frac{1}{2}}\right] = S_i := S(q_i),  
\end{align*}  
where the midpoint approximation for the cell average has been used, i.e.,  
\[\bar{q}_i=(1/\Delta x) \int_{I_i} q(x) \, dx:=q_i\quad \text{and}\quad\bar{S}_i=(1/\Delta x) \int_{I_i} S(q(x)) \, dx:=S_i.\] 
This is the Kurganov-Tadmor (KT) central scheme \cite{Kurganov2000}, where $F^*_{i-\frac{1}{2}}$ and $F^*_{i+\frac{1}{2}}$ are the functions $F(q)$ evaluated at the interfaces $x_i - \frac{1}{2} \Delta x$ and $x_i + \frac{1}{2} \Delta x$, respectively. The reconstruction schemes used in this paper involve two different reconstructions of the variables at the interfaces of each cell.  

\section*{WENO-Z}

WENO-Z is a fifth-order reconstruction scheme that uses a 5-point stencil to reconstruct values at each cell interface. For the cell $I_i$, it employs the stencil $S(x_i) = \{I_{i-2}, I_{i-1}, I_i, I_{i+1}, I_{i+2}\}$ centered at $x_i$ to reconstruct both $q_{i-\frac{1}{2}}^+$ and $q_{i+\frac{1}{2}}^-$. Here, $q_{i+\frac{1}{2}}^-$ represents the left state (upwind) reconstruction at $x_{i+\frac{1}{2}}$, while $q_{i-\frac{1}{2}}^+$ represents the right state (downwind) reconstruction at $x_{i-\frac{1}{2}}$, both computed from the stencil $S(x_i)$.

The left state reconstruction $q_{i+\frac{1}{2}}^-$ is obtained by dividing $S(x_i)$ into three substencils.
\begin{itemize}
    \item Left substencil: $S_1 = \{I_{i-2}, I_{i-1}, I_i\}$.
    \item Central substencil: $S_2 = \{I_{i-1}, I_i, I_{i+1}\}$.
    \item Right substencil: $S_3 = \{I_i, I_{i+1}, I_{i+2}\}$.
\end{itemize}
The reconstructed value is a convex combination of third-order interpolating polynomials
\begin{align*}
q_{i+\frac{1}{2}}^{(1)} &= \frac{1}{3}\bar{q}_{i-2} - \frac{7}{6}\bar{q}_{i-1} + \frac{11}{6}\bar{q}_i, \\
q_{i+\frac{1}{2}}^{(2)} &= -\frac{1}{6}\bar{q}_{i-1} + \frac{5}{6}\bar{q}_i + \frac{1}{3}\bar{q}_{i+1}, \\
q_{i+\frac{1}{2}}^{(3)} &= \frac{1}{3}\bar{q}_i + \frac{5}{6}\bar{q}_{i+1} - \frac{1}{6}\bar{q}_{i+2}.
\end{align*}

Using the midpoint approximation $\bar{q}_i = q_i$ for cell averages, the final reconstruction is
\begin{align*}
q_{i+\frac{1}{2}}^- = \sum_{k=1}^3 w_k q_{i+\frac{1}{2}}^{(k)},
\end{align*}
where the nonlinear weights $w_k$ are computed from the linear weights $d_k = \left(\frac{1}{10}, \frac{3}{5}, \frac{3}{10}\right)$ as
\begin{align*}
w_k = \frac{\alpha_k}{\sum_{l=1}^3 \alpha_l}, \quad
\alpha_k = d_k \left[1 + \left(\frac{\tau_5}{\epsilon + \beta_k}\right)^2\right], \quad\text{and}\quad
\tau_5 = |\beta_1 - \beta_3|.
\end{align*}
Here, $\epsilon \approx 10^{-40}\sim 10^{-6}$ (we used $10^{-40}$ for better results around discontinuities) prevents division by zero, and $\beta_k$ are the smoothness indicators
\bml
\label{Smooth_indicators}
\begin{align}
\beta_1 &= \frac{13}{12}(q_{i-2} - 2q_{i-1} + q_i)^2 + \frac{1}{4}(q_{i-2} - 4q_{i-1} + 3q_i)^2, \\
\beta_2 &= \frac{13}{12}(q_{i-1} - 2q_i + q_{i+1})^2 + \frac{1}{4}(q_{i-1} - q_{i+1})^2, \\
\beta_3 &= \frac{13}{12}(q_i - 2q_{i+1} + q_{i+2})^2 + \frac{1}{4}(3q_i - 4q_{i+1} + q_{i+2})^2.
\end{align}
\eml
The factor $\tau_5$ enhances the scheme's ability to select the smoothest solution near discontinuities while maintaining high accuracy in smooth regions.

For the right state reconstruction $q_{i-\frac{1}{2}}^+$, we use mirrored substencils (mirrored order)
\begin{align*}
S'_1 &= \{I_{i+2}, I_{i+1}, I_i\}, \\
S'_2 &= \{I_{i+1}, I_i, I_{i-1}\}, \\
S'_3 &= \{I_i, I_{i-1}, I_{i-2}\},
\end{align*}
yielding the polynomials
\begin{align*}
q_{i-\frac{1}{2}}^{(1)} &= \frac{1}{3}q_{i+2} - \frac{7}{6}q_{i+1} + \frac{11}{6}q_i, \\
q_{i-\frac{1}{2}}^{(2)} &= -\frac{1}{6}q_{i+1} + \frac{5}{6}q_i + \frac{1}{3}q_{i-1}, \\
q_{i-\frac{1}{2}}^{(3)} &= \frac{1}{3}q_i + \frac{5}{6}q_{i-1} - \frac{1}{6}q_{i-2}.
\end{align*}
The weights for the right state reconstruction reuse the smoothness indicators since $\beta'_1 = \beta_3$, $\beta'_2 = \beta_2$, and $\beta'_3 = \beta_1$, what yields
\begin{align*}
\alpha'_1 = d_1\left[1 + \left(\frac{\tau_5}{\epsilon + \beta_3}\right)^2\right], \quad
\alpha'_2 = \alpha_2, \quad\text{and}\quad
\alpha'_3 = d_3\left[1 + \left(\frac{\tau_5}{\epsilon + \beta_1}\right)^2\right].
\end{align*}

The flux reconstruction uses Lax-Friedrichs splitting \cite{Shu1997}
\begin{align*}
F^*_{i+\frac{1}{2}} = \frac{1}{2}\left[F(q^+_{i+\frac{1}{2}}) + F(q^-_{i+\frac{1}{2}}) - a_{i+\frac{1}{2}}(q^+_{i+\frac{1}{2}} - q^-_{i+\frac{1}{2}})\right],
\end{align*}
where $q^+_{i+\frac{1}{2}} = q^-_{i+\frac{3}{2}}$ is obtained from the stencil $S(x_{i+1})$. The parameter $a_{i\pm\frac{1}{2}}$ represents the maximum wave speed (typically the largest eigenvalue of $\partial F/\partial q$) and provides necessary numerical dissipation. For simplicity, we set $a_{i\pm\frac{1}{2}} = 1$, which maintains stability while improving computational efficiency.

\section*{MUSCL Scheme}

The MUSCL scheme \cite{VanLeer1974,VanLeer1977,VanLeer1979} reconstructs the fluxes at cell interfaces using a TVD approach. The interface fluxes are computed as
\begin{align*}
F^*_{i-1/2} &= \frac{1}{2}\left\{\left[F\left(q^R_{i-\frac{1}{2}}\right) + F\left(q^L_{i-\frac{1}{2}}\right)\right] - a_{i-\frac{1}{2}}\left(q^R_{i-\frac{1}{2}} - q^L_{i-\frac{1}{2}}\right)\right\}, \\
F^*_{i+1/2} &= \frac{1}{2}\left\{\left[F\left(q^R_{i+\frac{1}{2}}\right) + F\left(q^L_{i+\frac{1}{2}}\right)\right] - a_{i+\frac{1}{2}}\left(q^R_{i+\frac{1}{2}} - q^L_{i+\frac{1}{2}}\right)\right\},
\end{align*}
where $q^L$ and $q^R$ represent the left and right reconstructed states at each interface, and $a_{i\pm\frac{1}{2}}$ is the local maximum characteristic speed as aforementioned. We set it to 1 for simplicity. The interface states are reconstructed using slope-limited interpolation \cite{VanLeer1979}
\begin{align*}
q^L_{i+\frac{1}{2}} &= q_i + \frac{1}{2}\phi(r_i)(q_{i+1} - q_i), \\
q^R_{i+\frac{1}{2}} &= q_{i+1} - \frac{1}{2}\phi(r_{i+1})(q_{i+2} - q_{i+1}), \\
q^L_{i-\frac{1}{2}} &= q_{i-1} + \frac{1}{2}\phi(r_{i-1})(q_i - q_{i-1}), \\
q^R_{i-\frac{1}{2}} &= q_i - \frac{1}{2}\phi(r_i)(q_{i+1} - q_i),
\end{align*}
where the smoothness parameter $r_i$ is given by:
\begin{align*}
r_i = \frac{q_i - q_{i-1}}{q_{i+1} - q_i}.
\end{align*}
For our simulations, we employ the Sweby symmetric limiter \cite{Sweby1984}:
\begin{align*}
\phi(r) = \max\left[0, \min(\beta r, 1), \min(r, \beta)\right], \quad (1 \leq \beta \leq 2).
\end{align*}
Optimal results were obtained with $\beta = 1.5$ for both the relativistic Euler equations (tested separately) and the bulk Israel-Stewart equations. Smaller values of $\beta$ tend to excessively smooth the solutions, while larger values introduce spurious oscillations near discontinuities. The choice $\beta = 1.5$ provides a good balance between accuracy and stability for our applications.

\section*{Integrator and Numerical Differentiation}

For both schemes, we employed the second-order midpoint Runge-Kutta method \cite{Runge1895,Butcher2008} for time integration
\begin{align*}
q_i^{(s+1)} = q_i^{(s)} + \Delta t f\left(t_s + \frac{\Delta t}{2}, q_i^{(s)} + \frac{1}{2} f(t_s, q_i^{(s)})\right),
\end{align*}
where 
\[
f(t, q_i) = -\frac{1}{\Delta x}\left(F^*_{i+\frac{1}{2}} - F^*_{i-\frac{1}{2}}\right) + S(q_i).
\]

For spatial derivatives in the source terms, we used the 5th-order Central WENO (C-WENO) scheme \cite{Levy1999}, which maintains high accuracy while handling discontinuities and sharp gradients. The method utilizes the same three substencils $S_1$, $S_2$, and $S_3$ of $S(x_i)$ as before, with derivatives computed via
\begin{align*}
f'_1 &= \frac{f_{i-2} - 4f_{i-1} + 3f_i}{2\Delta x}, \\
f'_2 &= \frac{f_{i+1} - f_{i-1}}{2\Delta x}, \\
f'_3 &= \frac{-3f_i + 4f_{i+1} - f_{i+2}}{2\Delta x}.
\end{align*}
The final derivative combines these second order approximations using the smoothness indicators $\beta_k$ from \eqref{Smooth_indicators} as
\begin{align*}
f'(x_i) = \sum_{k=1}^3 w_k f'_k + \mathcal{O}(\Delta x^5),
\end{align*}
where order five is expected around smooth regions, with non-linear weights
\begin{align*}
w_k = \frac{\alpha_k}{\sum_{l=1}^3 \alpha_l} \quad \text{and}\quad
\alpha_k = \frac{c_k}{(\epsilon + \beta_k)^2},
\end{align*}
where $c_k = \left(\frac{1}{6}, \frac{2}{3}, \frac{1}{6}\right)$ are the optimal linear weights.

For the time derivative of $u$, we avoided numerical differentiation (particularly problematic at $t=0$) by instead solving
\[
\partial_t \Psi = (A^0)^{-1}\left(-A^1 \partial_x \Psi + B\right),
\]
yielding the explicit expression
\[
\partial_t u = -\frac{u^0 \left(P_e\tau_\Pi \partial_x\varepsilon + \tau_\Pi \partial_x\Pi - u \left\{\partial_x u \left[(P_e-1)\rho_\Pi\tau_\Pi + \zeta\right] + \Pi\right\}\right)}{u_0^2 \rho_\Pi \tau_\Pi - u^2 (P_e \rho_\Pi \tau_\Pi + \zeta)}.
\]

\section*{Boundary conditions}

The WENO-Z scheme requires 3 boundary points on each side, while MUSCL requires 2 boundary points per side. For a simulation domain of size $L$ discretized into $N$ grid points, the uniform grid spacing is given by $\Delta x = L/(N-1)$. The grid points are labeled as $i = 1, 2, \dots, N$, centered at positions $x_i$. The boundary regions consist of
the left boundary $x_1, \dots, x_{N_b}$ and the right boundary $x_{N-N_b+1}, \dots, x_N$, where $N_b = 3$ for WENO-Z and $N_b = 2$ for MUSCL. 

Unless otherwise specified, we implemented the following boundary conditions:
\begin{align*}
q_i &= q_{N_b+1} \quad \text{for } i = 1, \dots, N_b \quad \text{(left boundary)} \\
q_{N-N_b+i} &= q_{N-N_b} \quad \text{for } i = 1, \dots, N_b \quad \text{(right boundary)}
\end{align*}

\section*{Timestep}

We employed a Courant-Friedrichs-Lewy (CFL) condition \cite{courant1928original,courant1967translated} of $1/10$, setting $\Delta t = \Delta x/10$. Various grid spacings were tested, ranging from $\Delta x = 1/20$ downto $1/1000$, with the solutions showing consistent profiles across all resolutions. The primary differences observed were minor numerical artifacts that decreased with finer meshes. For computing the Rankine-Hugoniot conditions and Mach numbers, we used $\Delta x = 1/400$ to properly resolve the large gradients. While coarser grids produced similar results for these quantities, the finer resolution ensured better accuracy in gradient calculations.

\section*{Convergence Factor}

The convergence test followed the formula \cite{RezzollaBook}
\begin{align}
\label{p_converg}
\tilde{p}_a = \log_\gamma\left(\frac{E_a(h,h/\gamma)}{E_a(h/\gamma,h/\gamma^2)}\right),
\end{align}
where we compared solutions at successive grid spacings $h$, $h/\gamma$, and $h/\gamma^2$ (we used $\gamma = 2$). The average error between solutions was computed as
\[
E_a(h,k) = ||q_{a,i}^{(h)} - q_{a,i}^{(k)}||_1 \equiv \frac{1}{N}\sum_{i=1}^N \left|q_{a,i}^{(h)} - q_{a,i}^{(k)}\right|,
\]
where $q_{a,i}^{(h)}$ and $q_{a,i}^{(k)}$ represent the solution for variable $q_a$ at position $x_i$ and time $t$, computed with grid spacings $h$ and $k$ respectively. Assuming spatial independence of $\tilde{p}_a$, the error scales as $q_{a,i}^{(h)} - q_{a,i}^{(k)} = C k^{\tilde{p}_a} - C h^{\tilde{p}_a}$, where $C$ is a proportionality constant, leading to the convergence factor in \eqref{p_converg}. The convergence factor \(\tilde{p}_a \approx 1\) at discontinuities is expected for shock-capturing schemes and is a property of the global \(L_1\) norm. However, the fluid variables in the smooth upstream and downstream regions, from which the Rankine-Hugoniot conditions are calculated, are themselves much better resolved, contributing to the accuracy of the shock analysis.

From case 1 in Sec.\ \ref{subsec:case1}, Fig.~\ref{Fig_conv2} (a) and (b) show the convergence factors $\tilde{p}_a$ for WENO-Z scheme (higher in smooth regions) and MUSCL scheme, respectively. Both schemes maintain $\tilde{p}_a \approx 1$ around discontinuities, as expected for shock-capturing methods, show oscillatory behavior (physical behavior), and produce comparable results despite different approaches.
\begin{figure}[htb]
\centering
\includegraphics[width=8cm]{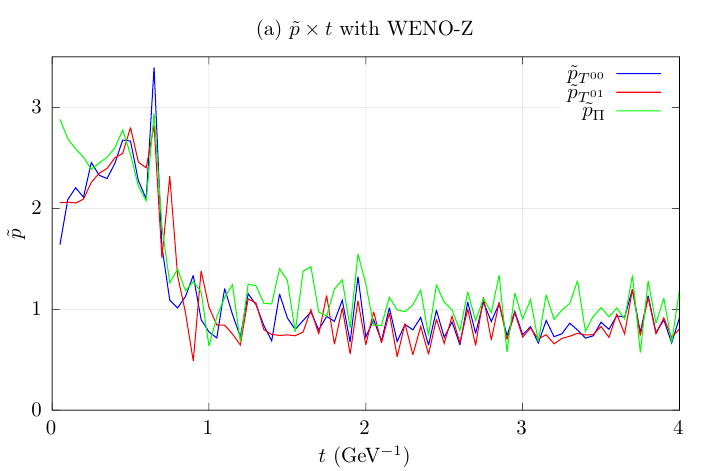}
\includegraphics[width=8cm]{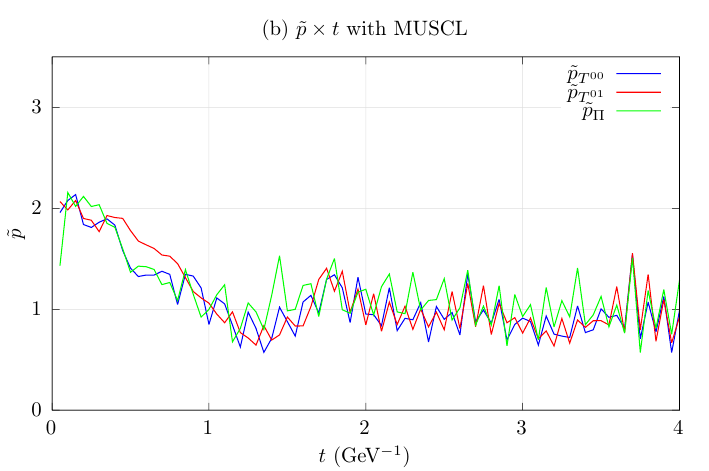}
\caption{\label{Fig_conv2} Convergence factors $\tilde{p}_a$ for (a) WENO-Z and (b) MUSCL schemes, showing expected $\tilde{p}_a \approx 1$ behavior with oscillations at discontinuities.}
\end{figure}

As for case 2 in Sec.\ \ref{subsec:case2}, both schemes maintain $\tilde{p}_a \approx 1$ at discontinuities, with oscillatory behavior, as expected (See Fig.\ \ref{Fig_conv3}). MUSCL performs better near shocks while WENO-Z excels in smooth solutions. However, the solutions are in accordance.
\begin{figure}[htb]
\centering
\includegraphics[width=8cm]{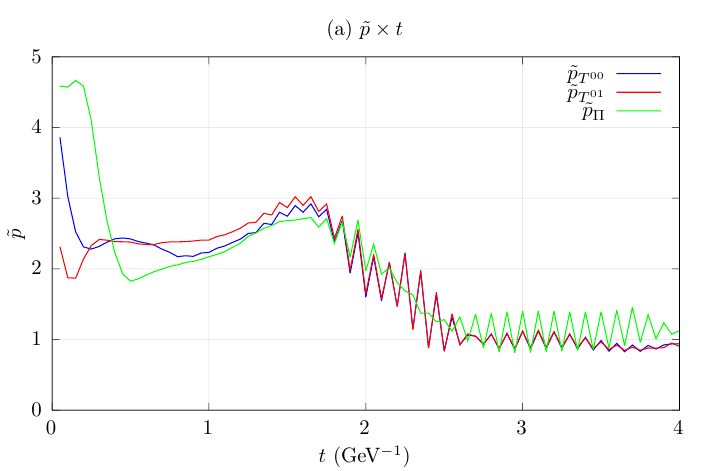}
\includegraphics[width=8cm]{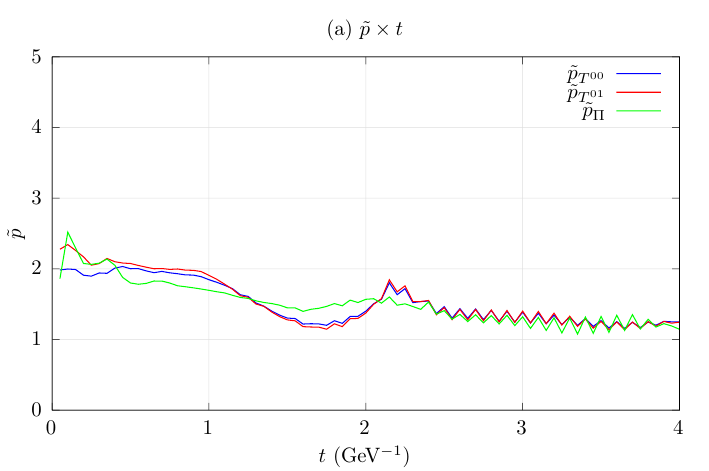}
\caption{\label{Fig_conv3} Convergence factors $\tilde{p}_a$ for (a) WENO-Z and (b) MUSCL schemes.}
\end{figure}

\section{Bjorken Flow}
\label{Bjorken}

A valuable benchmark for numerical codes in relativistic hydrodynamics is comparison with semi-analytical solutions like Bjorken flow \cite{Bjorken1983}. This solution describes a system that is invariant under Lorentz boosts along the beam axis (here taken to be the $x$-direction), meaning all physical quantities depend only on the proper time $\tau = \sqrt{t^2-x^2}$ and are independent of the spacetime rapidity $\eta_s = \frac{1}{2}\ln\left[(t+x)/(t-x)\right]$. In Milne coordinates, the Minkowski metric takes the form
\begin{align*}
ds^2 = -d\tau^2 + \tau^2 d\eta_s^2 + dy^2 + dz^2.
\end{align*}
The fluid four-velocity in these coordinates is $\tilde{u}^\mu = (1,0,0,0)$. 

Making the replacement $\partial_\mu \to \nabla_\mu$ in \eqref{Bulk_EOM}, and using $\nabla_\alpha \tilde{u}^\alpha = 1/\tau$, we obtain the simplified equations
\begin{align}
\label{Bulk_Milne}
\frac{d\varepsilon}{d\tau} &= -\frac{\varepsilon + P + \Pi}{\tau}, \\
\frac{d\Pi}{d\tau} &= -\frac{\Pi}{\tau_\Pi} - \frac{\zeta}{\tau \tau_\Pi},
\end{align}
where $\varepsilon = \varepsilon(\tau)$ and $\Pi = \Pi(\tau)$. This system of ordinary differential equations (ODEs) can be solved numerically with high precision using second-order Runge-Kutta methods. For our specific calculation, we assumed the same parametrization \eqref{Choices_Bulk}. For the initial conditions at $\tau_0 = 0.5$, we chose $\varepsilon(\tau_0) = 100\,\text{GeV}^4$ and $\Pi(\tau_0) = 0$. The solutions are shown in Fig. \ref{Fig3}.
\begin{figure}
\includegraphics[width=8cm]{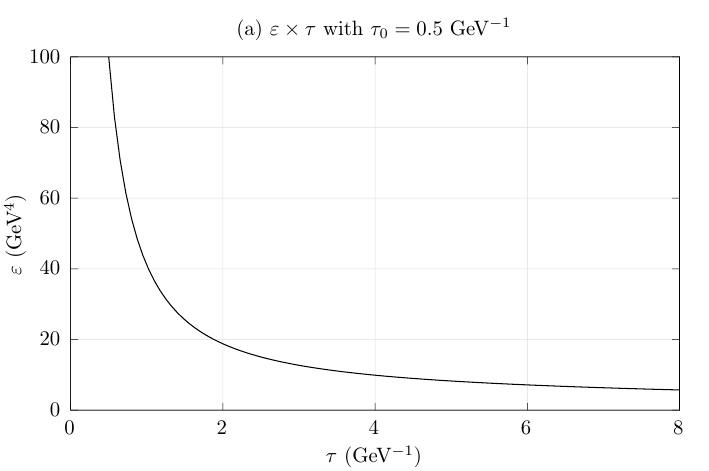}
\includegraphics[width=8cm]{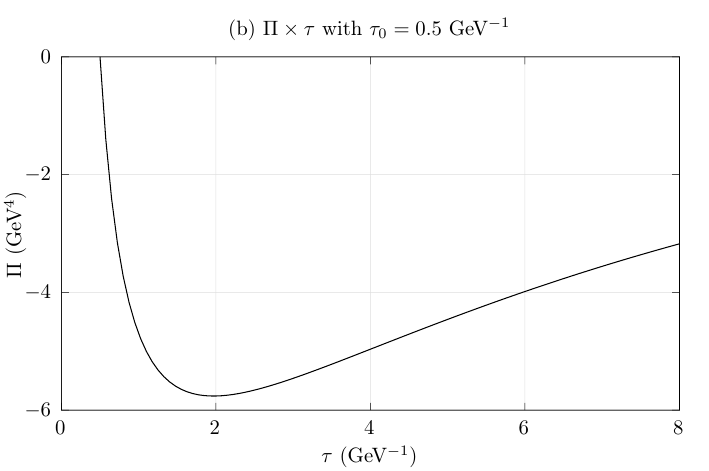}
\caption{\label{Fig3} Semi-analytical solutions of the Bjorken flow system \eqref{Bulk_Milne}: (a) Energy density $\varepsilon_{SA}$ versus $\tau$ and (b) Bulk viscous pressure $\Pi_{SA}$ versus $\tau$. Initial conditions: $\tau_0 = 0.5$, $\varepsilon(\tau_0) = 100$, $\Pi(\tau_0) = 0$. Solutions obtained using second-order Runge-Kutta with $\Delta \tau = 1/400$.}
\end{figure}

We now consider the same problem in Cartesian coordinates $(t,x)$, where the equations of motion \eqref{Bulk_EOM} apply directly. We solve \eqref{Bulk_EOM} numerically on the spatial interval $x \in [-5, 5]$. For initial conditions, we use the semi-analytical solution from Fig. \ref{Fig3} evaluated at $t_0 = \sqrt{0.5^2 + 5^2} \approx 5.02$, with energy density $\varepsilon(t_0,x) = \varepsilon_{SA}(\tau(t_0,x))$ and bulk viscous pressure $\Pi(t_0,x) = \Pi_{SA}(\tau(t_0,x))$, as shown in Fig. \ref{Fig4}(a) and (b) respectively. The flow velocity at $t=t_0$ in Cartesian coordinates is given by
\[
u^\mu(t_0,x) = \frac{t_0}{\tau_0}\left(1, \frac{x}{t_0}, 0, 0\right).
\]
\begin{figure}
\includegraphics[width=8cm]{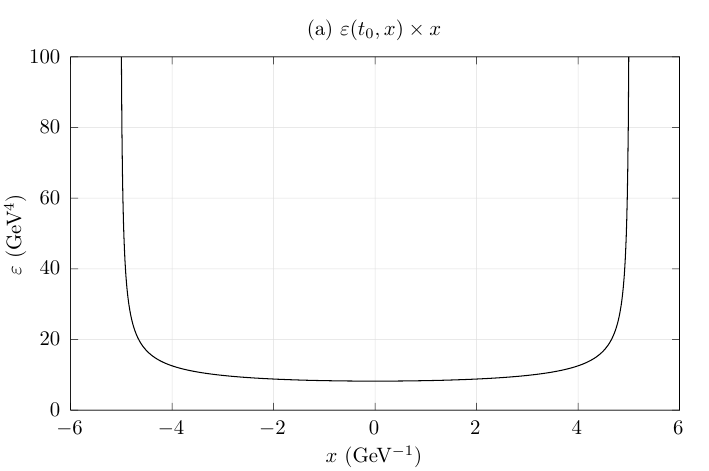}
\includegraphics[width=8cm]{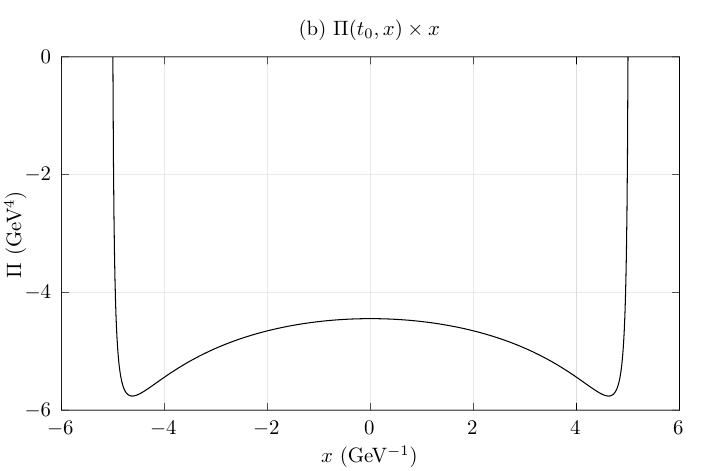}
\caption{\label{Fig4} Initial conditions at $t_0 \approx 5.02$ for (a) energy density $\varepsilon(t_0,x)$ and (b) bulk viscous pressure $\Pi(t_0,x)$, obtained from the semi-analytical Bjorken flow solution through $\tau(t_0,x) = \sqrt{t_0^2 - x^2}$.}
\end{figure}

The numerical solution of \eqref{Bulk_EOM} with these initial conditions presents special challenges. The fields $\varepsilon$ and $\Pi$ exhibit large gradients (See Fig.\ \ref{Fig4}) near the boundaries  $x_1,\ldots,x_{N_b}$ (left) and $x_{N-N_b+1},\ldots,x_N$ (right), where $N_b = 3$ for WENO-Z and $N_b = 2$ for MUSCL. Due to these steep gradients, we implemented modified boundary conditions:
\begin{enumerate}
    \item For the left boundary ($x_1,\ldots,x_{N_b}$), we used clamped cubic spline extrapolation based on the 8 adjacent points $x_{N_b+1},\ldots,x_{N_b+8}$.
    \item For the right boundary, we applied the symmetry relations $T^{00}(-x) = T^{00}(x)$, $T^{01}(-x) = -T^{01}(x)$, and $\Pi(-x) = \Pi(x)$ to mirror the left boundary solution.
\end{enumerate}
The clamped spline conditions were implemented using first-derivative constraints from the interior solution to maintain consistency with the physical behavior expected near the boundaries.

Figure \ref{Fig5} presents the numerical solution at $t = 8\,\text{GeV}^{-1}$ obtained using the WENO-Z scheme with $\Delta x = 1/400$ and $\Delta t = \Delta x/10$. Despite the large gradients near the boundaries, the numerical solution shows excellent agreement with the semi-analytical solutions $\varepsilon_{SA}(\tau(8,x))$ and $\Pi_{SA}(\tau(8,x))$. The relative error reaches machine precision throughout most of the domain, with only minor numerical oscillations visible in the profiles. As shown in Fig. \ref{Fig6}, the relative error remains below $6 \times 10^{-6}$ for the energy density $\varepsilon$ at the boundaries and below $2 \times 10^{-6}$ for the bulk viscous pressure $\Pi$ at the boundaries.
\begin{figure}[htb]
\includegraphics[width=8cm]{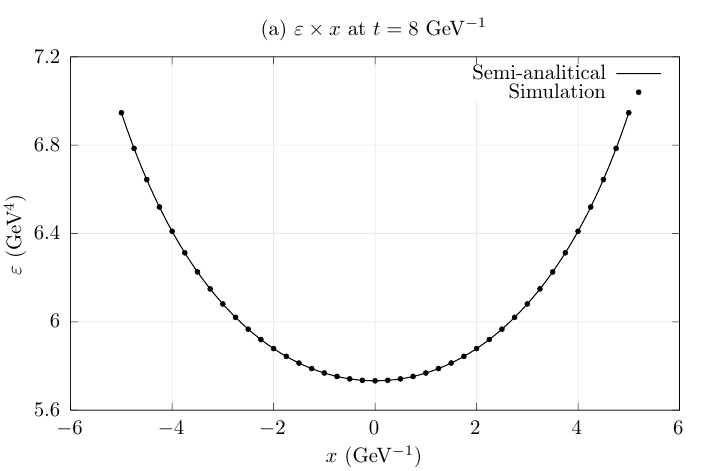}
\includegraphics[width=8cm]{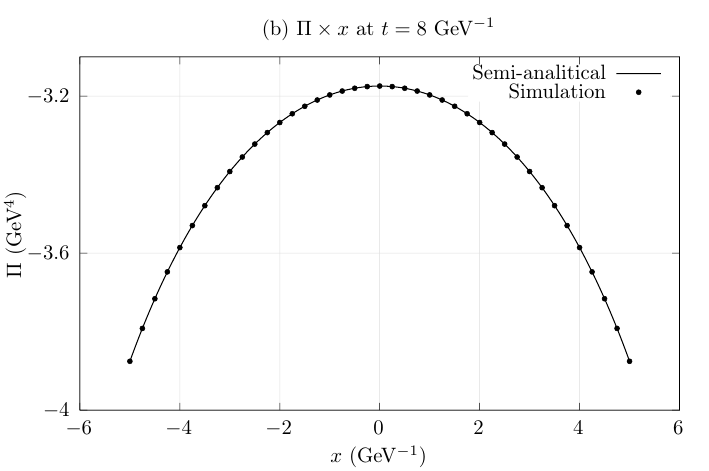}
\caption{\label{Fig5} Numerical results at $t = 8\,\text{GeV}^{-1}$: (a) Energy density $\varepsilon(x)$ and (b) bulk viscous pressure $\Pi(x)$. The semi-analytical solutions are shown as solid green curves, while the WENO-Z numerical results are displayed as data points.}
\end{figure}

\begin{figure}[htb]
\includegraphics[width=8cm]{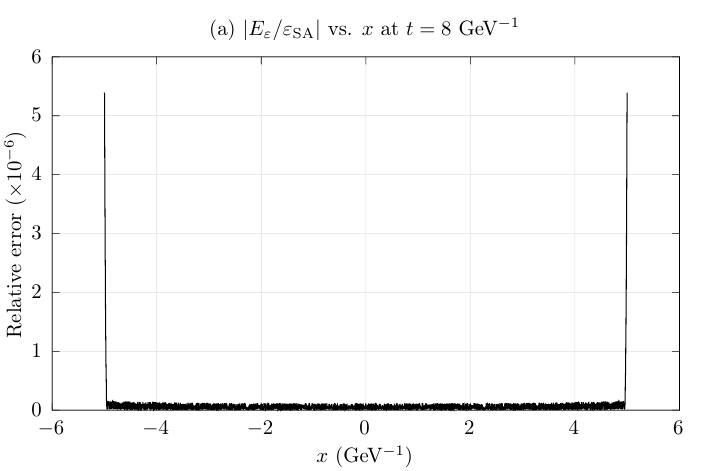}
\includegraphics[width=8cm]{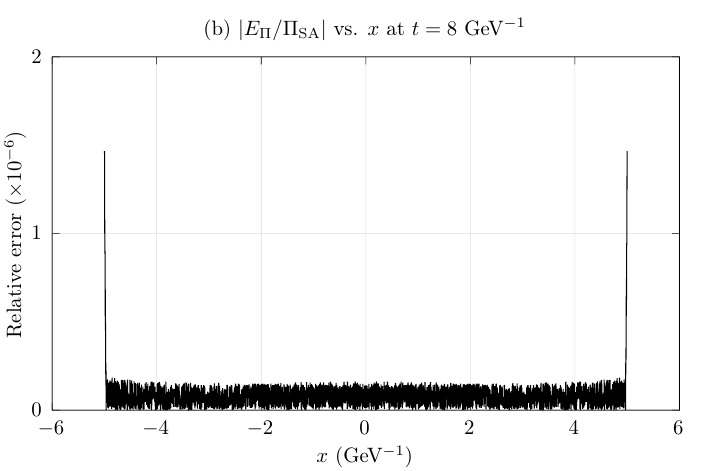}
\caption{\label{Fig6} Relative errors at $t = 8\,\text{GeV}^{-1}$ for (a) energy density and (b) bulk viscous pressure as functions of position $x$. The numerical solution achieves machine precision accuracy, with the observed residuals representing rounding errors inherent to floating-point arithmetic.} 
\end{figure}

%\bibliography{References}

\begin{thebibliography}{63}%
\makeatletter
\providecommand \@ifxundefined [1]{%
 \@ifx{#1\undefined}
}%
\providecommand \@ifnum [1]{%
 \ifnum #1\expandafter \@firstoftwo
 \else \expandafter \@secondoftwo
 \fi
}%
\providecommand \@ifx [1]{%
 \ifx #1\expandafter \@firstoftwo
 \else \expandafter \@secondoftwo
 \fi
}%
\providecommand \natexlab [1]{#1}%
\providecommand \enquote  [1]{``#1''}%
\providecommand \bibnamefont  [1]{#1}%
\providecommand \bibfnamefont [1]{#1}%
\providecommand \citenamefont [1]{#1}%
\providecommand \href@noop [0]{\@secondoftwo}%
\providecommand \href [0]{\begingroup \@sanitize@url \@href}%
\providecommand \@href[1]{\@@startlink{#1}\@@href}%
\providecommand \@@href[1]{\endgroup#1\@@endlink}%
\providecommand \@sanitize@url [0]{\catcode `\\12\catcode `\$12\catcode
  `\&12\catcode `\#12\catcode `\^12\catcode `\_12\catcode `\%12\relax}%
\providecommand \@@startlink[1]{}%
\providecommand \@@endlink[0]{}%
\providecommand \url  [0]{\begingroup\@sanitize@url \@url }%
\providecommand \@url [1]{\endgroup\@href {#1}{\urlprefix }}%
\providecommand \urlprefix  [0]{URL }%
\providecommand \Eprint [0]{\href }%
\providecommand \doibase [0]{https://doi.org/}%
\providecommand \selectlanguage [0]{\@gobble}%
\providecommand \bibinfo  [0]{\@secondoftwo}%
\providecommand \bibfield  [0]{\@secondoftwo}%
\providecommand \translation [1]{[#1]}%
\providecommand \BibitemOpen [0]{}%
\providecommand \bibitemStop [0]{}%
\providecommand \bibitemNoStop [0]{.\EOS\space}%
\providecommand \EOS [0]{\spacefactor3000\relax}%
\providecommand \BibitemShut  [1]{\csname bibitem#1\endcsname}%
\let\auto@bib@innerbib\@empty
%</preamble>
\bibitem [{\citenamefont {Challis}(1848)}]{Challis}%
  \BibitemOpen
  \bibfield  {author} {\bibinfo {author} {\bibfnamefont {J.}~\bibnamefont
  {Challis}},\ }\bibfield  {title} {\bibinfo {title} {{On the velocity of
  sound}},\ }\href@noop {} {\bibfield  {journal} {\bibinfo  {journal} {Philos.
  Magazine}\ }\textbf {\bibinfo {volume} {32}},\ \bibinfo {pages} {494}
  (\bibinfo {year} {1848})}\BibitemShut {NoStop}%
\bibitem [{\citenamefont {Koyama}\ \emph {et~al.}(1995)\citenamefont {Koyama},
  \citenamefont {Petre}, \citenamefont {Gotthelf}, \citenamefont {Hwang},
  \citenamefont {Matsuura}, \citenamefont {Ozaki},\ and\ \citenamefont
  {Holt}}]{Koyama:1995rr}%
  \BibitemOpen
  \bibfield  {author} {\bibinfo {author} {\bibfnamefont {K.}~\bibnamefont
  {Koyama}}, \bibinfo {author} {\bibfnamefont {R.}~\bibnamefont {Petre}},
  \bibinfo {author} {\bibfnamefont {E.~V.}\ \bibnamefont {Gotthelf}}, \bibinfo
  {author} {\bibfnamefont {U.}~\bibnamefont {Hwang}}, \bibinfo {author}
  {\bibfnamefont {M.}~\bibnamefont {Matsuura}}, \bibinfo {author}
  {\bibfnamefont {M.}~\bibnamefont {Ozaki}},\ and\ \bibinfo {author}
  {\bibfnamefont {S.~S.}\ \bibnamefont {Holt}},\ }\bibfield  {title} {\bibinfo
  {title} {{Evidence for shock acceleration of high-energy electrons in the
  supernova remnant SN1006}},\ }\href {https://doi.org/10.1038/378255a0}
  {\bibfield  {journal} {\bibinfo  {journal} {Nature}\ }\textbf {\bibinfo
  {volume} {378}},\ \bibinfo {pages} {255} (\bibinfo {year}
  {1995})}\BibitemShut {NoStop}%
\bibitem [{\citenamefont {Mekjian}(1978)}]{Mekjian:1978zz}%
  \BibitemOpen
  \bibfield  {author} {\bibinfo {author} {\bibfnamefont {A.~Z.}\ \bibnamefont
  {Mekjian}},\ }\bibfield  {title} {\bibinfo {title} {{Explosive
  nucleosynthesis, equilibrium thermodynamics, and relativistic heavy-ion
  collisions}},\ }\href {https://doi.org/10.1103/PhysRevC.17.1051} {\bibfield
  {journal} {\bibinfo  {journal} {Phys. Rev. C}\ }\textbf {\bibinfo {volume}
  {17}},\ \bibinfo {pages} {1051} (\bibinfo {year} {1978})}\BibitemShut
  {NoStop}%
\bibitem [{\citenamefont {Martí}(2019)}]{Marti2019}%
  \BibitemOpen
  \bibfield  {author} {\bibinfo {author} {\bibfnamefont {J.-M.}\ \bibnamefont
  {Martí}},\ }\bibfield  {title} {\bibinfo {title} {Numerical simulations of
  relativistic jets},\ }\href {https://doi.org/10.1007/s41115-019-0005-9}
  {\bibfield  {journal} {\bibinfo  {journal} {Living Reviews in Computational
  Astrophysics}\ }\textbf {\bibinfo {volume} {5}},\ \bibinfo {pages} {1}
  (\bibinfo {year} {2019})}\BibitemShut {NoStop}%
\bibitem [{\citenamefont {Blandford}\ \emph {et~al.}(2019)\citenamefont
  {Blandford}, \citenamefont {Meier},\ and\ \citenamefont
  {Readhead}}]{Blandford2019}%
  \BibitemOpen
  \bibfield  {author} {\bibinfo {author} {\bibfnamefont {R.}~\bibnamefont
  {Blandford}}, \bibinfo {author} {\bibfnamefont {D.}~\bibnamefont {Meier}},\
  and\ \bibinfo {author} {\bibfnamefont {A.}~\bibnamefont {Readhead}},\
  }\bibfield  {title} {\bibinfo {title} {Relativistic jets from active galactic
  nuclei},\ }\href {https://doi.org/10.1146/annurev-astro-081817-051948}
  {\bibfield  {journal} {\bibinfo  {journal} {Annual Review of Astronomy and
  Astrophysics}\ }\textbf {\bibinfo {volume} {57}},\ \bibinfo {pages} {467}
  (\bibinfo {year} {2019})}\BibitemShut {NoStop}%
\bibitem [{\citenamefont {Alinhac}(1999{\natexlab{a}})}]{Alinhac1999}%
  \BibitemOpen
  \bibfield  {author} {\bibinfo {author} {\bibfnamefont {S.}~\bibnamefont
  {Alinhac}},\ }\bibfield  {title} {\bibinfo {title} {Blowup of small data
  solutions for a quasilinear wave equation in two space dimensions},\ }\href
  {http://www.jstor.org/stable/121020} {\bibfield  {journal} {\bibinfo
  {journal} {Annals of Mathematics}\ }\textbf {\bibinfo {volume} {149}},\
  \bibinfo {pages} {97} (\bibinfo {year} {1999}{\natexlab{a}})}\BibitemShut
  {NoStop}%
\bibitem [{\citenamefont {Alinhac}(1999{\natexlab{b}})}]{Alinhac2001a}%
  \BibitemOpen
  \bibfield  {author} {\bibinfo {author} {\bibfnamefont {S.}~\bibnamefont
  {Alinhac}},\ }\bibfield  {title} {\bibinfo {title} {The null condition for
  quasilinear wave equations in two space dimensions {I}},\ }\href
  {https://doi.org/https://doi.org/10.1007/s002220100165} {\bibfield  {journal}
  {\bibinfo  {journal} {Invent. Math.}\ }\textbf {\bibinfo {volume} {145}},\
  \bibinfo {pages} {597} (\bibinfo {year} {1999}{\natexlab{b}})}\BibitemShut
  {NoStop}%
\bibitem [{\citenamefont {Alinhac}(2001)}]{Alinhac2001b}%
  \BibitemOpen
  \bibfield  {author} {\bibinfo {author} {\bibfnamefont {S.}~\bibnamefont
  {Alinhac}},\ }\bibfield  {title} {\bibinfo {title} {The null condition for
  quasilinear wave equations in two space dimensions. {II}},\ }\href@noop {}
  {\bibfield  {journal} {\bibinfo  {journal} {American Journal of Mathematics}\
  }\textbf {\bibinfo {volume} {123}},\ \bibinfo {pages} {1071} (\bibinfo {year}
  {2001})}\BibitemShut {NoStop}%
\bibitem [{\citenamefont {Christodoulou}(2007)}]{Christodoulou2007}%
  \BibitemOpen
  \bibfield  {author} {\bibinfo {author} {\bibfnamefont {D.}~\bibnamefont
  {Christodoulou}},\ }\href {https://doi.org/10.4171/031} {\emph {\bibinfo
  {title} {The formation of shocks in 3-dimensional fluids}}},\ edited by\
  \bibinfo {editor} {\bibfnamefont {E.~M.}\ \bibnamefont {in~Mathematics}},\
  Vol.\ \bibinfo {volume} {214}\ (\bibinfo  {publisher} {European Mathematical
  Society (EMS)},\ \bibinfo {year} {2007})\ pp.\ \bibinfo {pages}
  {viii+992}\BibitemShut {NoStop}%
\bibitem [{\citenamefont {Speck}(2016)}]{Speck}%
  \BibitemOpen
  \bibfield  {author} {\bibinfo {author} {\bibfnamefont {J.}~\bibnamefont
  {Speck}},\ }\href@noop {} {\emph {\bibinfo {title} {Shock Formation in
  Small-Data Solutions to 3D Quasilinear Wave Equations}}},\ edited by\
  \bibinfo {editor} {\bibfnamefont {M.}~\bibnamefont {Surveys}}\ and\ \bibinfo
  {editor} {\bibnamefont {Monographs}},\ Vol.\ \bibinfo {volume} {214}\
  (\bibinfo  {publisher} {American Mathematical Society (AMS)},\ \bibinfo
  {year} {2016})\ p.\ \bibinfo {pages} {515pp}\BibitemShut {NoStop}%
\bibitem [{\citenamefont {Disconzi}(2024)}]{Disconzi:2023rtt}%
  \BibitemOpen
  \bibfield  {author} {\bibinfo {author} {\bibfnamefont {M.~M.}\ \bibnamefont
  {Disconzi}},\ }\bibfield  {title} {\bibinfo {title} {{Recent developments in
  mathematical aspects of relativistic fluids}},\ }\href
  {https://doi.org/10.1007/s41114-024-00052-x} {\bibfield  {journal} {\bibinfo
  {journal} {Living Rev. Rel.}\ }\textbf {\bibinfo {volume} {27}},\ \bibinfo
  {pages} {6} (\bibinfo {year} {2024})},\ \Eprint
  {https://arxiv.org/abs/2308.09844} {arXiv:2308.09844 [math.AP]} \BibitemShut
  {NoStop}%
\bibitem [{\citenamefont {Romatschke}\ and\ \citenamefont
  {Romatschke}(2019)}]{Romatschke:2017ejr}%
  \BibitemOpen
  \bibfield  {author} {\bibinfo {author} {\bibfnamefont {P.}~\bibnamefont
  {Romatschke}}\ and\ \bibinfo {author} {\bibfnamefont {U.}~\bibnamefont
  {Romatschke}},\ }\href {https://doi.org/10.1017/9781108651998} {\emph
  {\bibinfo {title} {{Relativistic Fluid Dynamics In and Out of
  Equilibrium}}}},\ Cambridge Monographs on Mathematical Physics\ (\bibinfo
  {publisher} {Cambridge University Press},\ \bibinfo {year} {2019})\ \Eprint
  {https://arxiv.org/abs/1712.05815} {arXiv:1712.05815 [nucl-th]} \BibitemShut
  {NoStop}%
\bibitem [{\citenamefont {Heinz}\ and\ \citenamefont
  {Snellings}(2013)}]{Heinz:2013th}%
  \BibitemOpen
  \bibfield  {author} {\bibinfo {author} {\bibfnamefont {U.}~\bibnamefont
  {Heinz}}\ and\ \bibinfo {author} {\bibfnamefont {R.}~\bibnamefont
  {Snellings}},\ }\bibfield  {title} {\bibinfo {title} {{Collective flow and
  viscosity in relativistic heavy-ion collisions}},\ }\href
  {https://doi.org/10.1146/annurev-nucl-102212-170540} {\bibfield  {journal}
  {\bibinfo  {journal} {Ann. Rev. Nucl. Part. Sci.}\ }\textbf {\bibinfo
  {volume} {63}},\ \bibinfo {pages} {123} (\bibinfo {year} {2013})},\ \Eprint
  {https://arxiv.org/abs/1301.2826} {arXiv:1301.2826 [nucl-th]} \BibitemShut
  {NoStop}%
\bibitem [{\citenamefont {Alford}\ \emph {et~al.}(2018)\citenamefont {Alford},
  \citenamefont {Bovard}, \citenamefont {Hanauske}, \citenamefont {Rezzolla},\
  and\ \citenamefont {Schwenzer}}]{Alford:2017rxf}%
  \BibitemOpen
  \bibfield  {author} {\bibinfo {author} {\bibfnamefont {M.~G.}\ \bibnamefont
  {Alford}}, \bibinfo {author} {\bibfnamefont {L.}~\bibnamefont {Bovard}},
  \bibinfo {author} {\bibfnamefont {M.}~\bibnamefont {Hanauske}}, \bibinfo
  {author} {\bibfnamefont {L.}~\bibnamefont {Rezzolla}},\ and\ \bibinfo
  {author} {\bibfnamefont {K.}~\bibnamefont {Schwenzer}},\ }\bibfield  {title}
  {\bibinfo {title} {{Viscous Dissipation and Heat Conduction in Binary
  Neutron-Star Mergers}},\ }\href
  {https://doi.org/10.1103/PhysRevLett.120.041101} {\bibfield  {journal}
  {\bibinfo  {journal} {Phys. Rev. Lett.}\ }\textbf {\bibinfo {volume} {120}},\
  \bibinfo {pages} {041101} (\bibinfo {year} {2018})},\ \Eprint
  {https://arxiv.org/abs/1707.09475} {arXiv:1707.09475 [gr-qc]} \BibitemShut
  {NoStop}%
\bibitem [{\citenamefont {Shibata}\ and\ \citenamefont
  {Kiuchi}(2017)}]{Shibata:2017xht}%
  \BibitemOpen
  \bibfield  {author} {\bibinfo {author} {\bibfnamefont {M.}~\bibnamefont
  {Shibata}}\ and\ \bibinfo {author} {\bibfnamefont {K.}~\bibnamefont
  {Kiuchi}},\ }\bibfield  {title} {\bibinfo {title} {{Gravitational waves from
  remnant massive neutron stars of binary neutron star merger: Viscous
  hydrodynamics effects}},\ }\href {https://doi.org/10.1103/PhysRevD.95.123003}
  {\bibfield  {journal} {\bibinfo  {journal} {Phys. Rev. D}\ }\textbf {\bibinfo
  {volume} {95}},\ \bibinfo {pages} {123003} (\bibinfo {year} {2017})},\
  \Eprint {https://arxiv.org/abs/1705.06142} {arXiv:1705.06142 [astro-ph.HE]}
  \BibitemShut {NoStop}%
\bibitem [{\citenamefont {Shibata}\ \emph {et~al.}(2017)\citenamefont
  {Shibata}, \citenamefont {Kiuchi},\ and\ \citenamefont
  {Sekiguchi}}]{Shibata:2017jyf}%
  \BibitemOpen
  \bibfield  {author} {\bibinfo {author} {\bibfnamefont {M.}~\bibnamefont
  {Shibata}}, \bibinfo {author} {\bibfnamefont {K.}~\bibnamefont {Kiuchi}},\
  and\ \bibinfo {author} {\bibfnamefont {Y.-i.}\ \bibnamefont {Sekiguchi}},\
  }\bibfield  {title} {\bibinfo {title} {{General relativistic viscous
  hydrodynamics of differentially rotating neutron stars}},\ }\href
  {https://doi.org/10.1103/PhysRevD.95.083005} {\bibfield  {journal} {\bibinfo
  {journal} {Phys. Rev. D}\ }\textbf {\bibinfo {volume} {95}},\ \bibinfo
  {pages} {083005} (\bibinfo {year} {2017})},\ \Eprint
  {https://arxiv.org/abs/1703.10303} {arXiv:1703.10303 [astro-ph.HE]}
  \BibitemShut {NoStop}%
\bibitem [{\citenamefont {Most}\ \emph {et~al.}(2021)\citenamefont {Most},
  \citenamefont {Harris}, \citenamefont {Plumberg}, \citenamefont {Alford},
  \citenamefont {Noronha}, \citenamefont {Noronha-Hostler}, \citenamefont
  {Pretorius}, \citenamefont {Witek},\ and\ \citenamefont
  {Yunes}}]{Most:2021zvc}%
  \BibitemOpen
  \bibfield  {author} {\bibinfo {author} {\bibfnamefont {E.~R.}\ \bibnamefont
  {Most}}, \bibinfo {author} {\bibfnamefont {S.~P.}\ \bibnamefont {Harris}},
  \bibinfo {author} {\bibfnamefont {C.}~\bibnamefont {Plumberg}}, \bibinfo
  {author} {\bibfnamefont {M.~G.}\ \bibnamefont {Alford}}, \bibinfo {author}
  {\bibfnamefont {J.}~\bibnamefont {Noronha}}, \bibinfo {author} {\bibfnamefont
  {J.}~\bibnamefont {Noronha-Hostler}}, \bibinfo {author} {\bibfnamefont
  {F.}~\bibnamefont {Pretorius}}, \bibinfo {author} {\bibfnamefont
  {H.}~\bibnamefont {Witek}},\ and\ \bibinfo {author} {\bibfnamefont
  {N.}~\bibnamefont {Yunes}},\ }\bibfield  {title} {\bibinfo {title}
  {{Projecting the likely importance of weak-interaction-driven bulk viscosity
  in neutron star mergers}},\ }\href {https://doi.org/10.1093/mnras/stab2793}
  {\bibfield  {journal} {\bibinfo  {journal} {Mon. Not. Roy. Astron. Soc.}\
  }\textbf {\bibinfo {volume} {509}},\ \bibinfo {pages} {1096} (\bibinfo {year}
  {2021})},\ \Eprint {https://arxiv.org/abs/2107.05094} {arXiv:2107.05094
  [astro-ph.HE]} \BibitemShut {NoStop}%
\bibitem [{\citenamefont {Most}\ \emph {et~al.}(2024)\citenamefont {Most},
  \citenamefont {Haber}, \citenamefont {Harris}, \citenamefont {Zhang},
  \citenamefont {Alford},\ and\ \citenamefont {Noronha}}]{Most:2022yhe}%
  \BibitemOpen
  \bibfield  {author} {\bibinfo {author} {\bibfnamefont {E.~R.}\ \bibnamefont
  {Most}}, \bibinfo {author} {\bibfnamefont {A.}~\bibnamefont {Haber}},
  \bibinfo {author} {\bibfnamefont {S.~P.}\ \bibnamefont {Harris}}, \bibinfo
  {author} {\bibfnamefont {Z.}~\bibnamefont {Zhang}}, \bibinfo {author}
  {\bibfnamefont {M.~G.}\ \bibnamefont {Alford}},\ and\ \bibinfo {author}
  {\bibfnamefont {J.}~\bibnamefont {Noronha}},\ }\bibfield  {title} {\bibinfo
  {title} {{Emergence of Microphysical Bulk Viscosity in Binary Neutron Star
  Postmerger Dynamics}},\ }\href {https://doi.org/10.3847/2041-8213/ad454f}
  {\bibfield  {journal} {\bibinfo  {journal} {Astrophys. J. Lett.}\ }\textbf
  {\bibinfo {volume} {967}},\ \bibinfo {pages} {L14} (\bibinfo {year}
  {2024})},\ \Eprint {https://arxiv.org/abs/2207.00442} {arXiv:2207.00442
  [astro-ph.HE]} \BibitemShut {NoStop}%
\bibitem [{\citenamefont {Israel}(1976)}]{Israel:1976tn}%
  \BibitemOpen
  \bibfield  {author} {\bibinfo {author} {\bibfnamefont {W.}~\bibnamefont
  {Israel}},\ }\bibfield  {title} {\bibinfo {title} {{Nonstationary
  irreversible thermodynamics: A Causal relativistic theory}},\ }\href
  {https://doi.org/10.1016/0003-4916(76)90064-6} {\bibfield  {journal}
  {\bibinfo  {journal} {Annals Phys.}\ }\textbf {\bibinfo {volume} {100}},\
  \bibinfo {pages} {310} (\bibinfo {year} {1976})}\BibitemShut {NoStop}%
\bibitem [{\citenamefont {Israel}\ and\ \citenamefont
  {Stewart}(1976)}]{Israel:1976efz}%
  \BibitemOpen
  \bibfield  {author} {\bibinfo {author} {\bibfnamefont {W.}~\bibnamefont
  {Israel}}\ and\ \bibinfo {author} {\bibfnamefont {J.~M.}\ \bibnamefont
  {Stewart}},\ }\bibfield  {title} {\bibinfo {title} {{Thermodynamics of
  nonstationary and transient effects in a relativistic gas}},\ }\href
  {https://doi.org/10.1016/0375-9601(76)90075-X} {\bibfield  {journal}
  {\bibinfo  {journal} {Phys. Lett. A}\ }\textbf {\bibinfo {volume} {58}},\
  \bibinfo {pages} {213} (\bibinfo {year} {1976})}\BibitemShut {NoStop}%
\bibitem [{\citenamefont {Israel}\ and\ \citenamefont
  {Stewart}(1979)}]{Israel1979}%
  \BibitemOpen
  \bibfield  {author} {\bibinfo {author} {\bibfnamefont {W.}~\bibnamefont
  {Israel}}\ and\ \bibinfo {author} {\bibfnamefont {J.~M.}\ \bibnamefont
  {Stewart}},\ }\bibfield  {title} {\bibinfo {title} {{On transient
  relativistic thermodynamics and kinetic theory. II}},\ }\href
  {https://doi.org/10.1098/rspa.1979.0005} {\bibfield  {journal} {\bibinfo
  {journal} {Proc. R. Soc. A}\ }\textbf {\bibinfo {volume} {365}},\ \bibinfo
  {pages} {43} (\bibinfo {year} {1979})}\BibitemShut {NoStop}%
\bibitem [{\citenamefont {Baier}\ \emph {et~al.}(2008)\citenamefont {Baier},
  \citenamefont {Romatschke}, \citenamefont {Son}, \citenamefont {Starinets},\
  and\ \citenamefont {Stephanov}}]{Baier:2007ix}%
  \BibitemOpen
  \bibfield  {author} {\bibinfo {author} {\bibfnamefont {R.}~\bibnamefont
  {Baier}}, \bibinfo {author} {\bibfnamefont {P.}~\bibnamefont {Romatschke}},
  \bibinfo {author} {\bibfnamefont {D.~T.}\ \bibnamefont {Son}}, \bibinfo
  {author} {\bibfnamefont {A.~O.}\ \bibnamefont {Starinets}},\ and\ \bibinfo
  {author} {\bibfnamefont {M.~A.}\ \bibnamefont {Stephanov}},\ }\bibfield
  {title} {\bibinfo {title} {{Relativistic viscous hydrodynamics, conformal
  invariance, and holography}},\ }\href
  {https://doi.org/10.1088/1126-6708/2008/04/100} {\bibfield  {journal}
  {\bibinfo  {journal} {JHEP}\ }\textbf {\bibinfo {volume} {04}},\ \bibinfo
  {pages} {100}},\ \Eprint {https://arxiv.org/abs/0712.2451} {arXiv:0712.2451
  [hep-th]} \BibitemShut {NoStop}%
\bibitem [{\citenamefont {Denicol}\ \emph {et~al.}(2012)\citenamefont
  {Denicol}, \citenamefont {Niemi}, \citenamefont {Molnar},\ and\ \citenamefont
  {Rischke}}]{Denicol:2012cn}%
  \BibitemOpen
  \bibfield  {author} {\bibinfo {author} {\bibfnamefont {G.~S.}\ \bibnamefont
  {Denicol}}, \bibinfo {author} {\bibfnamefont {H.}~\bibnamefont {Niemi}},
  \bibinfo {author} {\bibfnamefont {E.}~\bibnamefont {Molnar}},\ and\ \bibinfo
  {author} {\bibfnamefont {D.~H.}\ \bibnamefont {Rischke}},\ }\bibfield
  {title} {\bibinfo {title} {{Derivation of transient relativistic fluid
  dynamics from the Boltzmann equation}},\ }\href
  {https://doi.org/10.1103/PhysRevD.85.114047} {\bibfield  {journal} {\bibinfo
  {journal} {Phys. Rev. D}\ }\textbf {\bibinfo {volume} {85}},\ \bibinfo
  {pages} {114047} (\bibinfo {year} {2012})},\ \bibinfo {note} {[Erratum:
  Phys.Rev.D 91, 039902 (2015)]},\ \Eprint {https://arxiv.org/abs/1202.4551}
  {arXiv:1202.4551 [nucl-th]} \BibitemShut {NoStop}%
\bibitem [{\citenamefont {Hiscock}\ and\ \citenamefont
  {Lindblom}(1983)}]{Hiscock:1983zz}%
  \BibitemOpen
  \bibfield  {author} {\bibinfo {author} {\bibfnamefont {W.~A.}\ \bibnamefont
  {Hiscock}}\ and\ \bibinfo {author} {\bibfnamefont {L.}~\bibnamefont
  {Lindblom}},\ }\bibfield  {title} {\bibinfo {title} {{Stability and causality
  in dissipative relativistic fluids}},\ }\href
  {https://doi.org/10.1016/0003-4916(83)90288-9} {\bibfield  {journal}
  {\bibinfo  {journal} {Annals Phys.}\ }\textbf {\bibinfo {volume} {151}},\
  \bibinfo {pages} {466} (\bibinfo {year} {1983})}\BibitemShut {NoStop}%
\bibitem [{\citenamefont {Olson}(1990)}]{Olson:1990rzl}%
  \BibitemOpen
  \bibfield  {author} {\bibinfo {author} {\bibfnamefont {T.~S.}\ \bibnamefont
  {Olson}},\ }\bibfield  {title} {\bibinfo {title} {Stability and causality in
  the israel-stewart energy frame theory},\ }\href
  {https://doi.org/10.1016/0003-4916(90)90366-V} {\bibfield  {journal}
  {\bibinfo  {journal} {Annals Phys.}\ }\textbf {\bibinfo {volume} {199}},\
  \bibinfo {pages} {18} (\bibinfo {year} {1990})}\BibitemShut {NoStop}%
\bibitem [{\citenamefont {Hiscock}\ and\ \citenamefont
  {Lindblom}(1985)}]{Hiscock:1985zz}%
  \BibitemOpen
  \bibfield  {author} {\bibinfo {author} {\bibfnamefont {W.~A.}\ \bibnamefont
  {Hiscock}}\ and\ \bibinfo {author} {\bibfnamefont {L.}~\bibnamefont
  {Lindblom}},\ }\bibfield  {title} {\bibinfo {title} {{Generic instabilities
  in first-order dissipative relativistic fluid theories}},\ }\href
  {https://doi.org/10.1103/PhysRevD.31.725} {\bibfield  {journal} {\bibinfo
  {journal} {Phys. Rev. D}\ }\textbf {\bibinfo {volume} {31}},\ \bibinfo
  {pages} {725} (\bibinfo {year} {1985})}\BibitemShut {NoStop}%
\bibitem [{\citenamefont {Bemfica}\ \emph {et~al.}(2021)\citenamefont
  {Bemfica}, \citenamefont {Disconzi}, \citenamefont {Hoang}, \citenamefont
  {Noronha},\ and\ \citenamefont {Radosz}}]{Bemfica:2020xym}%
  \BibitemOpen
  \bibfield  {author} {\bibinfo {author} {\bibfnamefont {F.~S.}\ \bibnamefont
  {Bemfica}}, \bibinfo {author} {\bibfnamefont {M.~M.}\ \bibnamefont
  {Disconzi}}, \bibinfo {author} {\bibfnamefont {V.}~\bibnamefont {Hoang}},
  \bibinfo {author} {\bibfnamefont {J.}~\bibnamefont {Noronha}},\ and\ \bibinfo
  {author} {\bibfnamefont {M.}~\bibnamefont {Radosz}},\ }\bibfield  {title}
  {\bibinfo {title} {{Nonlinear Constraints on Relativistic Fluids Far from
  Equilibrium}},\ }\href {https://doi.org/10.1103/PhysRevLett.126.222301}
  {\bibfield  {journal} {\bibinfo  {journal} {Phys. Rev. Lett.}\ }\textbf
  {\bibinfo {volume} {126}},\ \bibinfo {pages} {222301} (\bibinfo {year}
  {2021})},\ \Eprint {https://arxiv.org/abs/2005.11632} {arXiv:2005.11632
  [hep-th]} \BibitemShut {NoStop}%
\bibitem [{\citenamefont {Bemfica}\ \emph
  {et~al.}(2019{\natexlab{a}})\citenamefont {Bemfica}, \citenamefont
  {Disconzi},\ and\ \citenamefont {Noronha}}]{Bemfica:2019cop}%
  \BibitemOpen
  \bibfield  {author} {\bibinfo {author} {\bibfnamefont {F.~S.}\ \bibnamefont
  {Bemfica}}, \bibinfo {author} {\bibfnamefont {M.~M.}\ \bibnamefont
  {Disconzi}},\ and\ \bibinfo {author} {\bibfnamefont {J.}~\bibnamefont
  {Noronha}},\ }\bibfield  {title} {\bibinfo {title} {{Causality of the
  Einstein-Israel-Stewart Theory with Bulk Viscosity}},\ }\href
  {https://doi.org/10.1103/PhysRevLett.122.221602} {\bibfield  {journal}
  {\bibinfo  {journal} {Phys. Rev. Lett.}\ }\textbf {\bibinfo {volume} {122}},\
  \bibinfo {pages} {221602} (\bibinfo {year} {2019}{\natexlab{a}})},\ \Eprint
  {https://arxiv.org/abs/1901.06701} {arXiv:1901.06701 [gr-qc]} \BibitemShut
  {NoStop}%
\bibitem [{\citenamefont {Olson}\ and\ \citenamefont
  {Hiscock}(1990)}]{OlsonHiscock1990}%
  \BibitemOpen
  \bibfield  {author} {\bibinfo {author} {\bibfnamefont {T.~S.}\ \bibnamefont
  {Olson}}\ and\ \bibinfo {author} {\bibfnamefont {W.~A.}\ \bibnamefont
  {Hiscock}},\ }\bibfield  {title} {\bibinfo {title} {Stability and causality
  in dissipative relativistic fluids},\ }\href@noop {} {\bibfield  {journal}
  {\bibinfo  {journal} {Phys. Rev. D}\ }\textbf {\bibinfo {volume} {41}},\
  \bibinfo {pages} {3687} (\bibinfo {year} {1990})}\BibitemShut {NoStop}%
\bibitem [{\citenamefont {Geroch}\ and\ \citenamefont
  {Lindblom}(1991)}]{GEROCH1991394}%
  \BibitemOpen
  \bibfield  {author} {\bibinfo {author} {\bibfnamefont {R.}~\bibnamefont
  {Geroch}}\ and\ \bibinfo {author} {\bibfnamefont {L.}~\bibnamefont
  {Lindblom}},\ }\bibfield  {title} {\bibinfo {title} {Causal theories of
  dissipative relativistic fluids},\ }\href
  {https://doi.org/https://doi.org/10.1016/0003-4916(91)90063-E} {\bibfield
  {journal} {\bibinfo  {journal} {Annals of Physics}\ }\textbf {\bibinfo
  {volume} {207}},\ \bibinfo {pages} {394} (\bibinfo {year}
  {1991})}\BibitemShut {NoStop}%
\bibitem [{\citenamefont {Disconzi}\ \emph {et~al.}(2020)\citenamefont
  {Disconzi}, \citenamefont {Hoang},\ and\ \citenamefont
  {Radosz}}]{Disconzi:2020ijk}%
  \BibitemOpen
  \bibfield  {author} {\bibinfo {author} {\bibfnamefont {M.~M.}\ \bibnamefont
  {Disconzi}}, \bibinfo {author} {\bibfnamefont {V.}~\bibnamefont {Hoang}},\
  and\ \bibinfo {author} {\bibfnamefont {M.}~\bibnamefont {Radosz}},\
  }\bibfield  {title} {\bibinfo {title} {{Breakdown of smooth solutions to the
  M\"uller-Israel-Stewart equations of relativistic viscous fluids}}\ }\href
  {https://doi.org/10.1007/s11005-023-01677-9} {10.1007/s11005-023-01677-9}
  (\bibinfo {year} {2020}),\ \Eprint {https://arxiv.org/abs/2008.03841}
  {arXiv:2008.03841 [math.AP]} \BibitemShut {NoStop}%
\bibitem [{\citenamefont {B\"arlin}(2023)}]{BARLIN2023103901}%
  \BibitemOpen
  \bibfield  {author} {\bibinfo {author} {\bibfnamefont {J.}~\bibnamefont
  {B\"arlin}},\ }\bibfield  {title} {\bibinfo {title} {Formation of
  singularities in solutions to nonlinear hyperbolic systems with general
  sources},\ }\href
  {https://doi.org/https://doi.org/10.1016/j.nonrwa.2023.103901} {\bibfield
  {journal} {\bibinfo  {journal} {Nonlinear Analysis: Real World Applications}\
  }\textbf {\bibinfo {volume} {73}},\ \bibinfo {pages} {103901} (\bibinfo
  {year} {2023})}\BibitemShut {NoStop}%
\bibitem [{\citenamefont {Landau}\ and\ \citenamefont
  {Lifshitz}(1987)}]{landau1987fluid}%
  \BibitemOpen
  \bibfield  {author} {\bibinfo {author} {\bibfnamefont {L.~D.}\ \bibnamefont
  {Landau}}\ and\ \bibinfo {author} {\bibfnamefont {E.~M.}\ \bibnamefont
  {Lifshitz}},\ }\href@noop {} {\emph {\bibinfo {title} {Fluid Mechanics:
  Volume 6}}},\ \bibinfo {edition} {2nd}\ ed.,\ Course of Theoretical Physics\
  (\bibinfo  {publisher} {Butterworth-Heinemann},\ \bibinfo {address}
  {Oxford},\ \bibinfo {year} {1987})\BibitemShut {NoStop}%
\bibitem [{\citenamefont {Farlow}(1993)}]{farlow1993partial}%
  \BibitemOpen
  \bibfield  {author} {\bibinfo {author} {\bibfnamefont {S.~J.}\ \bibnamefont
  {Farlow}},\ }\href@noop {} {\emph {\bibinfo {title} {Partial Differential
  Equations for Scientists and Engineers}}}\ (\bibinfo  {publisher} {Dover
  Publications},\ \bibinfo {address} {New York},\ \bibinfo {year}
  {1993})\BibitemShut {NoStop}%
\bibitem [{\citenamefont {Lax}(1964)}]{Lax1964}%
  \BibitemOpen
  \bibfield  {author} {\bibinfo {author} {\bibfnamefont {P.~D.}\ \bibnamefont
  {Lax}},\ }\bibfield  {title} {\bibinfo {title} {Development of singularities
  of solutions of nonlinear hyperbolic partial differential equations},\ }\href
  {https://doi.org/10.1063/1.1704154} {\bibfield  {journal} {\bibinfo
  {journal} {Journal of Mathematical Physics}\ }\textbf {\bibinfo {volume}
  {5}},\ \bibinfo {pages} {611} (\bibinfo {year} {1964})}\BibitemShut {NoStop}%
\bibitem [{\citenamefont {John}(1974)}]{John:1974}%
  \BibitemOpen
  \bibfield  {author} {\bibinfo {author} {\bibfnamefont {F.}~\bibnamefont
  {John}},\ }\bibfield  {title} {\bibinfo {title} {Formation of singularities
  in one-dimensional nonlinear wave propagation},\ }\href
  {https://doi.org/https://doi.org/10.1002/cpa.3160270307} {\bibfield
  {journal} {\bibinfo  {journal} {Communications on Pure and Applied
  Mathematics}\ }\textbf {\bibinfo {volume} {27}},\ \bibinfo {pages} {377}
  (\bibinfo {year} {1974})},\ \Eprint
  {https://arxiv.org/abs/https://onlinelibrary.wiley.com/doi/pdf/10.1002/cpa.3160270307}
  {https://onlinelibrary.wiley.com/doi/pdf/10.1002/cpa.3160270307} \BibitemShut
  {NoStop}%
\bibitem [{\citenamefont {Liu}(1979)}]{LIU197992}%
  \BibitemOpen
  \bibfield  {author} {\bibinfo {author} {\bibfnamefont {T.-P.}\ \bibnamefont
  {Liu}},\ }\bibfield  {title} {\bibinfo {title} {Development of singularities
  in the nonlinear waves for quasi-linear hyperbolic partial differential
  equations},\ }\href
  {https://doi.org/https://doi.org/10.1016/0022-0396(79)90082-2} {\bibfield
  {journal} {\bibinfo  {journal} {Journal of Differential Equations}\ }\textbf
  {\bibinfo {volume} {33}},\ \bibinfo {pages} {92} (\bibinfo {year}
  {1979})}\BibitemShut {NoStop}%
\bibitem [{\citenamefont {Gavassino}\ and\ \citenamefont
  {Noronha}(2024)}]{Gavassino:2023xkt}%
  \BibitemOpen
  \bibfield  {author} {\bibinfo {author} {\bibfnamefont {L.}~\bibnamefont
  {Gavassino}}\ and\ \bibinfo {author} {\bibfnamefont {J.}~\bibnamefont
  {Noronha}},\ }\bibfield  {title} {\bibinfo {title} {{Relativistic bulk
  rheology: From neutron star mergers to viscous cosmology}},\ }\href
  {https://doi.org/10.1103/PhysRevD.109.096040} {\bibfield  {journal} {\bibinfo
   {journal} {Phys. Rev. D}\ }\textbf {\bibinfo {volume} {109}},\ \bibinfo
  {pages} {096040} (\bibinfo {year} {2024})},\ \Eprint
  {https://arxiv.org/abs/2305.04119} {arXiv:2305.04119 [gr-qc]} \BibitemShut
  {NoStop}%
\bibitem [{\citenamefont {Borsanyi}\ \emph {et~al.}(2010)\citenamefont
  {Borsanyi}, \citenamefont {Endrodi}, \citenamefont {Fodor}, \citenamefont
  {Jakovac}, \citenamefont {Katz}, \citenamefont {Krieg}, \citenamefont
  {Ratti},\ and\ \citenamefont {Szabo}}]{Borsanyi:2010cj}%
  \BibitemOpen
  \bibfield  {author} {\bibinfo {author} {\bibfnamefont {S.}~\bibnamefont
  {Borsanyi}}, \bibinfo {author} {\bibfnamefont {G.}~\bibnamefont {Endrodi}},
  \bibinfo {author} {\bibfnamefont {Z.}~\bibnamefont {Fodor}}, \bibinfo
  {author} {\bibfnamefont {A.}~\bibnamefont {Jakovac}}, \bibinfo {author}
  {\bibfnamefont {S.~D.}\ \bibnamefont {Katz}}, \bibinfo {author}
  {\bibfnamefont {S.}~\bibnamefont {Krieg}}, \bibinfo {author} {\bibfnamefont
  {C.}~\bibnamefont {Ratti}},\ and\ \bibinfo {author} {\bibfnamefont {K.~K.}\
  \bibnamefont {Szabo}},\ }\bibfield  {title} {\bibinfo {title} {{The QCD
  equation of state with dynamical quarks}},\ }\href
  {https://doi.org/10.1007/JHEP11(2010)077} {\bibfield  {journal} {\bibinfo
  {journal} {JHEP}\ }\textbf {\bibinfo {volume} {11}},\ \bibinfo {pages}
  {077}},\ \Eprint {https://arxiv.org/abs/1007.2580} {arXiv:1007.2580
  [hep-lat]} \BibitemShut {NoStop}%
\bibitem [{\citenamefont {Bemfica}(2025)}]{Bemfica_code}%
  \BibitemOpen
  \bibfield  {author} {\bibinfo {author} {\bibfnamefont {F.~S.}\ \bibnamefont
  {Bemfica}},\ }\href {https://doi.org/10.5281/zenodo.16118344} {\bibinfo
  {title} {{Code for IS with bulk viscosity}}} (\bibinfo {year} {2025}),\
  \bibinfo {note} {\doi{10.5281/zenodo.16118344}}\BibitemShut {NoStop}%
\bibitem [{\citenamefont {Rezzolla}\ and\ \citenamefont
  {Zanotti}(2013)}]{RezzollaBook}%
  \BibitemOpen
  \bibfield  {author} {\bibinfo {author} {\bibfnamefont {L.}~\bibnamefont
  {Rezzolla}}\ and\ \bibinfo {author} {\bibfnamefont {O.}~\bibnamefont
  {Zanotti}},\ }\href@noop {} {\emph {\bibinfo {title} {Relativistic
  Hydrodynamics}}}\ (\bibinfo  {publisher} {Oxford University Press},\ \bibinfo
  {year} {2013})\BibitemShut {NoStop}%
\bibitem [{\citenamefont {Bemfica}\ \emph {et~al.}(2018)\citenamefont
  {Bemfica}, \citenamefont {Disconzi},\ and\ \citenamefont
  {Noronha}}]{Bemfica:2017wps}%
  \BibitemOpen
  \bibfield  {author} {\bibinfo {author} {\bibfnamefont {F.~S.}\ \bibnamefont
  {Bemfica}}, \bibinfo {author} {\bibfnamefont {M.~M.}\ \bibnamefont
  {Disconzi}},\ and\ \bibinfo {author} {\bibfnamefont {J.}~\bibnamefont
  {Noronha}},\ }\bibfield  {title} {\bibinfo {title} {{Causality and existence
  of solutions of relativistic viscous fluid dynamics with gravity}},\ }\href
  {https://doi.org/10.1103/PhysRevD.98.104064} {\bibfield  {journal} {\bibinfo
  {journal} {Phys. Rev. D}\ }\textbf {\bibinfo {volume} {98}},\ \bibinfo
  {pages} {104064} (\bibinfo {year} {2018})},\ \Eprint
  {https://arxiv.org/abs/1708.06255} {arXiv:1708.06255 [gr-qc]} \BibitemShut
  {NoStop}%
\bibitem [{\citenamefont {Bemfica}\ \emph
  {et~al.}(2019{\natexlab{b}})\citenamefont {Bemfica}, , \citenamefont
  {Disconzi},\ and\ \citenamefont {Noronha}}]{Bemfica:2019knx}%
  \BibitemOpen
  \bibfield  {author} {\bibinfo {author} {\bibfnamefont {F.~S.}\ \bibnamefont
  {Bemfica}}, , \bibinfo {author} {\bibfnamefont {M.~M.}\ \bibnamefont
  {Disconzi}},\ and\ \bibinfo {author} {\bibfnamefont {J.}~\bibnamefont
  {Noronha}},\ }\bibfield  {title} {\bibinfo {title} {{Nonlinear Causality of
  General First-Order Relativistic Viscous Hydrodynamics}},\ }\href
  {https://doi.org/10.1103/PhysRevD.100.104020} {\bibfield  {journal} {\bibinfo
   {journal} {Phys. Rev. D}\ }\textbf {\bibinfo {volume} {100}},\ \bibinfo
  {pages} {104020} (\bibinfo {year} {2019}{\natexlab{b}})},\ \bibinfo {note}
  {[Erratum: Phys.Rev.D 105, 069902 (2022)]},\ \Eprint
  {https://arxiv.org/abs/1907.12695} {arXiv:1907.12695 [gr-qc]} \BibitemShut
  {NoStop}%
\bibitem [{\citenamefont {Bemfica}\ \emph {et~al.}(2022)\citenamefont
  {Bemfica}, \citenamefont {Disconzi},\ and\ \citenamefont
  {Noronha}}]{Bemfica:2020zjp}%
  \BibitemOpen
  \bibfield  {author} {\bibinfo {author} {\bibfnamefont {F.~S.}\ \bibnamefont
  {Bemfica}}, \bibinfo {author} {\bibfnamefont {M.~M.}\ \bibnamefont
  {Disconzi}},\ and\ \bibinfo {author} {\bibfnamefont {J.}~\bibnamefont
  {Noronha}},\ }\bibfield  {title} {\bibinfo {title} {{First-Order
  General-Relativistic Viscous Fluid Dynamics}},\ }\href
  {https://doi.org/10.1103/PhysRevX.12.021044} {\bibfield  {journal} {\bibinfo
  {journal} {Phys. Rev. X}\ }\textbf {\bibinfo {volume} {12}},\ \bibinfo
  {pages} {021044} (\bibinfo {year} {2022})},\ \Eprint
  {https://arxiv.org/abs/2009.11388} {arXiv:2009.11388 [gr-qc]} \BibitemShut
  {NoStop}%
\bibitem [{\citenamefont {Kovtun}(2019)}]{Kovtun:2019hdm}%
  \BibitemOpen
  \bibfield  {author} {\bibinfo {author} {\bibfnamefont {P.}~\bibnamefont
  {Kovtun}},\ }\bibfield  {title} {\bibinfo {title} {{First-order relativistic
  hydrodynamics is stable}},\ }\href {https://doi.org/10.1007/JHEP10(2019)034}
  {\bibfield  {journal} {\bibinfo  {journal} {JHEP}\ }\textbf {\bibinfo
  {volume} {10}},\ \bibinfo {pages} {034}},\ \Eprint
  {https://arxiv.org/abs/1907.08191} {arXiv:1907.08191 [hep-th]} \BibitemShut
  {NoStop}%
\bibitem [{\citenamefont {Hoult}\ and\ \citenamefont
  {Kovtun}(2020)}]{Kovtun:2020eho}%
  \BibitemOpen
  \bibfield  {author} {\bibinfo {author} {\bibfnamefont {R.~E.}\ \bibnamefont
  {Hoult}}\ and\ \bibinfo {author} {\bibfnamefont {P.}~\bibnamefont {Kovtun}},\
  }\bibfield  {title} {\bibinfo {title} {{Stable and causal relativistic
  Navier-Stokes equations}},\ }\href {https://doi.org/10.1007/JHEP06(2020)067}
  {\bibfield  {journal} {\bibinfo  {journal} {JHEP}\ }\textbf {\bibinfo
  {volume} {06}},\ \bibinfo {pages} {067}},\ \Eprint
  {https://arxiv.org/abs/2004.04102} {arXiv:2004.04102 [hep-th]} \BibitemShut
  {NoStop}%
\bibitem [{\citenamefont {Freistuhler}(2021)}]{Freistuhler:2021lla}%
  \BibitemOpen
  \bibfield  {author} {\bibinfo {author} {\bibfnamefont {H.}~\bibnamefont
  {Freistuhler}},\ }\bibfield  {title} {\bibinfo {title} {{Nonexistence and
  existence of shock profiles in the Bemfica-Disconzi-Noronha model}},\ }\href
  {https://doi.org/10.1103/PhysRevD.103.124045} {\bibfield  {journal} {\bibinfo
   {journal} {Phys. Rev. D}\ }\textbf {\bibinfo {volume} {103}},\ \bibinfo
  {pages} {124045} (\bibinfo {year} {2021})},\ \Eprint
  {https://arxiv.org/abs/2103.16661} {arXiv:2103.16661 [math.AP]} \BibitemShut
  {NoStop}%
\bibitem [{\citenamefont {Pandya}\ and\ \citenamefont
  {Pretorius}(2021)}]{Pandya:2021ief}%
  \BibitemOpen
  \bibfield  {author} {\bibinfo {author} {\bibfnamefont {A.}~\bibnamefont
  {Pandya}}\ and\ \bibinfo {author} {\bibfnamefont {F.}~\bibnamefont
  {Pretorius}},\ }\bibfield  {title} {\bibinfo {title} {{Numerical exploration
  of first-order relativistic hydrodynamics}},\ }\href
  {https://doi.org/10.1103/PhysRevD.104.023015} {\bibfield  {journal} {\bibinfo
   {journal} {Phys. Rev. D}\ }\textbf {\bibinfo {volume} {104}},\ \bibinfo
  {pages} {023015} (\bibinfo {year} {2021})},\ \Eprint
  {https://arxiv.org/abs/2104.00804} {arXiv:2104.00804 [gr-qc]} \BibitemShut
  {NoStop}%
\bibitem [{\citenamefont {Pandya}\ \emph {et~al.}(2022)\citenamefont {Pandya},
  \citenamefont {Most},\ and\ \citenamefont {Pretorius}}]{Pandya:2022pif}%
  \BibitemOpen
  \bibfield  {author} {\bibinfo {author} {\bibfnamefont {A.}~\bibnamefont
  {Pandya}}, \bibinfo {author} {\bibfnamefont {E.~R.}\ \bibnamefont {Most}},\
  and\ \bibinfo {author} {\bibfnamefont {F.}~\bibnamefont {Pretorius}},\
  }\bibfield  {title} {\bibinfo {title} {{Conservative finite volume scheme for
  first-order viscous relativistic hydrodynamics}},\ }\href
  {https://doi.org/10.1103/PhysRevD.105.123001} {\bibfield  {journal} {\bibinfo
   {journal} {Phys. Rev. D}\ }\textbf {\bibinfo {volume} {105}},\ \bibinfo
  {pages} {123001} (\bibinfo {year} {2022})},\ \Eprint
  {https://arxiv.org/abs/2201.12317} {arXiv:2201.12317 [gr-qc]} \BibitemShut
  {NoStop}%
\bibitem [{\citenamefont {Bea}\ \emph {et~al.}(2021)\citenamefont {Bea},
  \citenamefont {Casalderrey-Solana}, \citenamefont {Giannakopoulos},
  \citenamefont {Jansen}, \citenamefont {Krippendorf}, \citenamefont {Mateos},
  \citenamefont {Sanchez-Garitaonandia},\ and\ \citenamefont
  {Zilh{\~a}o}}]{Bea:2021zol}%
  \BibitemOpen
  \bibfield  {author} {\bibinfo {author} {\bibfnamefont {Y.}~\bibnamefont
  {Bea}}, \bibinfo {author} {\bibfnamefont {J.}~\bibnamefont
  {Casalderrey-Solana}}, \bibinfo {author} {\bibfnamefont {T.}~\bibnamefont
  {Giannakopoulos}}, \bibinfo {author} {\bibfnamefont {A.}~\bibnamefont
  {Jansen}}, \bibinfo {author} {\bibfnamefont {S.}~\bibnamefont {Krippendorf}},
  \bibinfo {author} {\bibfnamefont {D.}~\bibnamefont {Mateos}}, \bibinfo
  {author} {\bibfnamefont {M.}~\bibnamefont {Sanchez-Garitaonandia}},\ and\
  \bibinfo {author} {\bibfnamefont {M.}~\bibnamefont {Zilh{\~a}o}},\ }\bibfield
   {title} {\bibinfo {title} {{Spinodal Gravitational Waves}},\ }\href@noop {}
  {\  (\bibinfo {year} {2021})},\ \Eprint {https://arxiv.org/abs/2112.15478}
  {arXiv:2112.15478 [hep-th]} \BibitemShut {NoStop}%
\bibitem [{\citenamefont {Van~Leer}(1979)}]{VanLeer1979}%
  \BibitemOpen
  \bibfield  {author} {\bibinfo {author} {\bibfnamefont {B.}~\bibnamefont
  {Van~Leer}},\ }\bibfield  {title} {\bibinfo {title} {Towards the ultimate
  conservative difference scheme. {V}. a second-order sequel to godunov's
  method},\ }\href {https://doi.org/10.1016/0021-9991(79)90145-1} {\bibfield
  {journal} {\bibinfo  {journal} {Journal of Computational Physics}\ }\textbf
  {\bibinfo {volume} {32}},\ \bibinfo {pages} {101} (\bibinfo {year}
  {1979})}\BibitemShut {NoStop}%
\bibitem [{\citenamefont {Borges}\ \emph {et~al.}(2008)\citenamefont {Borges},
  \citenamefont {Carmona}, \citenamefont {Costa},\ and\ \citenamefont
  {Don}}]{Borges2008}%
  \BibitemOpen
  \bibfield  {author} {\bibinfo {author} {\bibfnamefont {R.}~\bibnamefont
  {Borges}}, \bibinfo {author} {\bibfnamefont {M.}~\bibnamefont {Carmona}},
  \bibinfo {author} {\bibfnamefont {B.}~\bibnamefont {Costa}},\ and\ \bibinfo
  {author} {\bibfnamefont {W.-S.}\ \bibnamefont {Don}},\ }\bibfield  {title}
  {\bibinfo {title} {An improved weighted essentially non-oscillatory scheme
  for hyperbolic conservation laws},\ }\href
  {https://doi.org/10.1016/j.jcp.2007.11.038} {\bibfield  {journal} {\bibinfo
  {journal} {Journal of Computational Physics}\ }\textbf {\bibinfo {volume}
  {227}},\ \bibinfo {pages} {3191} (\bibinfo {year} {2008})}\BibitemShut
  {NoStop}%
\bibitem [{\citenamefont {Kurganov}\ and\ \citenamefont
  {Tadmor}(2000)}]{Kurganov2000}%
  \BibitemOpen
  \bibfield  {author} {\bibinfo {author} {\bibfnamefont {A.}~\bibnamefont
  {Kurganov}}\ and\ \bibinfo {author} {\bibfnamefont {E.}~\bibnamefont
  {Tadmor}},\ }\bibfield  {title} {\bibinfo {title} {New high-resolution
  central schemes for nonlinear conservation laws and convection-diffusion
  equations},\ }\href {https://doi.org/10.1006/jcph.2000.6459} {\bibfield
  {journal} {\bibinfo  {journal} {Journal of Computational Physics}\ }\textbf
  {\bibinfo {volume} {160}},\ \bibinfo {pages} {241} (\bibinfo {year}
  {2000})}\BibitemShut {NoStop}%
\bibitem [{\citenamefont {Shu}(1997)}]{Shu1997}%
  \BibitemOpen
  \bibfield  {author} {\bibinfo {author} {\bibfnamefont {C.-W.}\ \bibnamefont
  {Shu}},\ }\bibfield  {title} {\bibinfo {title} {Essentially non-oscillatory
  and weighted essentially non-oscillatory schemes for hyperbolic conservation
  laws},\ }\href@noop {} {\bibfield  {journal} {\bibinfo  {journal}
  {NASA/CR-97-206253}\ } (\bibinfo {year} {1997})},\ \bibinfo {note} {{ICASE
  Report No. 97-65}}\BibitemShut {NoStop}%
\bibitem [{\citenamefont {Van~Leer}(1974)}]{VanLeer1974}%
  \BibitemOpen
  \bibfield  {author} {\bibinfo {author} {\bibfnamefont {B.}~\bibnamefont
  {Van~Leer}},\ }\bibfield  {title} {\bibinfo {title} {Towards the ultimate
  conservative difference scheme. {II}. monotonicity and conservation combined
  in a second-order scheme},\ }\href
  {https://doi.org/10.1016/0021-9991(74)90019-9} {\bibfield  {journal}
  {\bibinfo  {journal} {Journal of Computational Physics}\ }\textbf {\bibinfo
  {volume} {14}},\ \bibinfo {pages} {361} (\bibinfo {year} {1974})}\BibitemShut
  {NoStop}%
\bibitem [{\citenamefont {Van~Leer}(1977)}]{VanLeer1977}%
  \BibitemOpen
  \bibfield  {author} {\bibinfo {author} {\bibfnamefont {B.}~\bibnamefont
  {Van~Leer}},\ }\bibfield  {title} {\bibinfo {title} {Towards the ultimate
  conservative difference scheme. {III}. upstream-centered finite-difference
  schemes for ideal compressible flow},\ }\href
  {https://doi.org/10.1016/0021-9991(77)90094-8} {\bibfield  {journal}
  {\bibinfo  {journal} {Journal of Computational Physics}\ }\textbf {\bibinfo
  {volume} {23}},\ \bibinfo {pages} {263} (\bibinfo {year} {1977})}\BibitemShut
  {NoStop}%
\bibitem [{\citenamefont {Sweby}(1984)}]{Sweby1984}%
  \BibitemOpen
  \bibfield  {author} {\bibinfo {author} {\bibfnamefont {P.~K.}\ \bibnamefont
  {Sweby}},\ }\bibfield  {title} {\bibinfo {title} {High resolution schemes
  using flux limiters for hyperbolic conservation laws},\ }\href
  {https://doi.org/10.1137/0721062} {\bibfield  {journal} {\bibinfo  {journal}
  {SIAM Journal on Numerical Analysis}\ }\textbf {\bibinfo {volume} {21}},\
  \bibinfo {pages} {995} (\bibinfo {year} {1984})},\ \Eprint
  {https://arxiv.org/abs/https://doi.org/10.1137/0721062}
  {https://doi.org/10.1137/0721062} \BibitemShut {NoStop}%
\bibitem [{\citenamefont {Runge}(1895)}]{Runge1895}%
  \BibitemOpen
  \bibfield  {author} {\bibinfo {author} {\bibfnamefont {C.}~\bibnamefont
  {Runge}},\ }\bibfield  {title} {\bibinfo {title} {Über die numerische
  auflösung von differentialgleichungen},\ }\href
  {https://doi.org/10.1007/BF01446807} {\bibfield  {journal} {\bibinfo
  {journal} {Mathematische Annalen}\ }\textbf {\bibinfo {volume} {46}},\
  \bibinfo {pages} {167} (\bibinfo {year} {1895})}\BibitemShut {NoStop}%
\bibitem [{\citenamefont {Butcher}(2008)}]{Butcher2008}%
  \BibitemOpen
  \bibfield  {author} {\bibinfo {author} {\bibfnamefont {J.~C.}\ \bibnamefont
  {Butcher}},\ }\href@noop {} {\emph {\bibinfo {title} {Numerical Methods for
  Ordinary Differential Equations}}},\ \bibinfo {edition} {2nd}\ ed.\ (\bibinfo
   {publisher} {Wiley},\ \bibinfo {year} {2008})\ pp.\ \bibinfo {pages}
  {98--102}\BibitemShut {NoStop}%
\bibitem [{\citenamefont {Levy}\ \emph {et~al.}(1999)\citenamefont {Levy},
  \citenamefont {Puppo},\ and\ \citenamefont {Russo}}]{Levy1999}%
  \BibitemOpen
  \bibfield  {author} {\bibinfo {author} {\bibfnamefont {D.}~\bibnamefont
  {Levy}}, \bibinfo {author} {\bibfnamefont {G.}~\bibnamefont {Puppo}},\ and\
  \bibinfo {author} {\bibfnamefont {G.}~\bibnamefont {Russo}},\ }\bibfield
  {title} {\bibinfo {title} {Central weno schemes for hyperbolic systems of
  conservation laws},\ }\href {https://doi.org/10.1051/m2an:1999158} {\bibfield
   {journal} {\bibinfo  {journal} {ESAIM: Mathematical Modelling and Numerical
  Analysis}\ }\textbf {\bibinfo {volume} {33}},\ \bibinfo {pages} {547}
  (\bibinfo {year} {1999})}\BibitemShut {NoStop}%
\bibitem [{\citenamefont {Courant}\ \emph {et~al.}(1928)\citenamefont
  {Courant}, \citenamefont {Friedrichs},\ and\ \citenamefont
  {Lewy}}]{courant1928original}%
  \BibitemOpen
  \bibfield  {author} {\bibinfo {author} {\bibfnamefont {R.}~\bibnamefont
  {Courant}}, \bibinfo {author} {\bibfnamefont {K.}~\bibnamefont
  {Friedrichs}},\ and\ \bibinfo {author} {\bibfnamefont {H.}~\bibnamefont
  {Lewy}},\ }\bibfield  {title} {\bibinfo {title} {\"uber die partiellen
  differenzengleichungen der mathematischen physik},\ }\href
  {https://doi.org/10.1007/BF01448839} {\bibfield  {journal} {\bibinfo
  {journal} {Mathematische Annalen}\ }\textbf {\bibinfo {volume} {100}},\
  \bibinfo {pages} {32} (\bibinfo {year} {1928})}\BibitemShut {NoStop}%
\bibitem [{\citenamefont {Courant}\ \emph {et~al.}(1967)\citenamefont
  {Courant}, \citenamefont {Friedrichs},\ and\ \citenamefont
  {Lewy}}]{courant1967translated}%
  \BibitemOpen
  \bibfield  {author} {\bibinfo {author} {\bibfnamefont {R.}~\bibnamefont
  {Courant}}, \bibinfo {author} {\bibfnamefont {K.}~\bibnamefont
  {Friedrichs}},\ and\ \bibinfo {author} {\bibfnamefont {H.}~\bibnamefont
  {Lewy}},\ }\bibfield  {title} {\bibinfo {title} {On the partial difference
  equations of mathematical physics},\ }\href
  {https://doi.org/10.1147/rd.112.0215} {\bibfield  {journal} {\bibinfo
  {journal} {IBM Journal of Research and Development}\ }\textbf {\bibinfo
  {volume} {11}},\ \bibinfo {pages} {215} (\bibinfo {year} {1967})},\ \bibinfo
  {note} {english translation of the original 1928 paper.}\BibitemShut {Stop}%
\bibitem [{\citenamefont {Bjorken}(1983)}]{Bjorken1983}%
  \BibitemOpen
  \bibfield  {author} {\bibinfo {author} {\bibfnamefont {J.~D.}\ \bibnamefont
  {Bjorken}},\ }\bibfield  {title} {\bibinfo {title} {Highly relativistic
  nucleus-nucleus collisions: The central rapidity region},\ }\href
  {https://doi.org/10.1103/PhysRevD.27.140} {\bibfield  {journal} {\bibinfo
  {journal} {Physical Review D}\ }\textbf {\bibinfo {volume} {27}},\ \bibinfo
  {pages} {140} (\bibinfo {year} {1983})}\BibitemShut {NoStop}%
\end{thebibliography}


\end{document}